\documentclass[reqno,12pt]{article}
\usepackage{indentfirst}
\usepackage{fouriernc}
\usepackage{anyfontsize}
\usepackage{ntheorem}
\usepackage{amsmath,amssymb}
\usepackage[doublespacing]{setspace}
\usepackage{sectsty}
\usepackage[rgb]{xcolor}
\usepackage[norule,bottom]{footmisc}
\usepackage[justification=centering,textfont={sc},labelfont={rm}]{caption}
\usepackage{varioref}
\usepackage{graphicx, caption, subcaption,epstopdf,secdot}
\usepackage{comment, natbib}
\usepackage{tikz}
\usepackage[normalem]{ulem}
\usepackage[colorlinks=true,linkcolor=blue]{hyperref}
\hypersetup{colorlinks=true,citecolor=blue}
\usepackage{tkz-euclide}
\usepackage[top=1in, bottom=1in, left=1in, right=1in]{geometry}

\theoremindent\parindent
\makeatletter    
\renewtheoremstyle{plain}{\item[\hskip\labelsep\hskip-\parindent \theorem@headerfont ##1\ ##2\theorem@separator]}{\item[\hskip\labelsep\hskip-\parindent \theorem@headerfont ##1\ ##2\ (##3)\theorem@separator]}
\makeatother
\theoremheaderfont{\scshape}
\theorembodyfont{\rmfamily}
\theoremseparator{.}

\newtheorem{claim}{Claim}
\newtheorem{observation}{Observation}

\newtheorem{corollary}{Corollary}

\newtheorem{lemma}{Lemma}
\newtheorem{proposition}{Proposition}

\newtheorem{innercustomthm}{Proposition}
\newenvironment{customthm}[1]
  {\renewcommand\theinnercustomthm{#1}\innercustomthm}
  {\endinnercustomthm}

\newenvironment{proof}[1][Proof]{\indent\textit{#1.}}{\ \rule{0.5em}{0.5em}}
\sectionfont{\normalfont\scshape\centering}
\subsectionfont{\itshape}
\makeatletter
\def\@biblabel#1{\hspace*{-\labelsep}}

\makeatother
\makeatletter
\renewcommand\@makefnmark{\mbox{\textsuperscript{\normalfont\@thefnmark}}}
\renewcommand\@makefntext[1]{\indent\makebox[2.5em][r]{\@thefnmark.\,}#1}
\if@titlepage
  \renewenvironment{abstract}{      \titlepage
      \null\vfil
      \@beginparpenalty\@lowpenalty
      \begin{center}        \@endparpenalty\@M
      \end{center}}     {\par\vfil\null\endtitlepage}
\else
  \renewenvironment{abstract}{      \if@twocolumn
      \else
        \small
        \begin{center}        \end{center}      \fi}
      {\if@twocolumn\else\endquotation\fi}
\fi
\renewcommand\thetable{\@Roman\c@table}

\renewcommand\thefigure{\@Roman\c@figure}

\newcommand{\E}{\mathbb{E}}
\newcommand{\V}{\mathbb{V}}
\newcommand{\DT}{{\Delta\Theta}}
\newcommand{\RP}{\mathcal{R}}
\newcommand{\real}{\mathbb{R}}

\newcommand{\dd}{\text{d}}
\newcommand{\ddd}{\text{ d}}
\newcommand{\supp}{\text{supp}}

\newcommand{\z}{z}
\newcommand{\co}{\mbox{co}}
\newcommand{\B}{{\mathcal B}}
\newcommand{\lt}{{L}}
\newcommand{\rt}{{R}}
\newcommand{\down}{{D}}
\newcommand{\longsquiggly}{\xymatrix{{}\ar@{~>}[r]&{}}}

\newcommand{\quiz}[1]{\mbox{{\fontfamily{qag}\selectfont {\normalsize #1}}}}
\newcommand{\quizT}{\quiz{true}{}}
\newcommand{\quizF}{\quiz{false}{}}
\newcommand{\quizU}{\quiz{uncertain}{}}

\DeclareMathOperator*{\argmax}{argmax}

\makeatother
\labelformat{equation}{(#1)}
\begin{document}
\title{\vspace{-2cm}{\bf Optimal Attention Management:\\ A Tractable Framework}\thanks{%
We would like to thank David Ahn, Nemanja Antic, 
Andreas Blume, Ben Brooks, Andrew Caplin, Sylvain Chassang, Piotr Dworczak, Haluk Ergin, Brett Green, Yingni Guo, Filip Mat\v{e}jka, David Pearce, Doron Ravid, Chris Shannon, and Philipp Strack for their feedback, as well as audiences at NASMES (Davis), ES Summer School (Singapore), Stony Brook Summer Festival, UC Berkeley, University of Arizona, ASU, and Northwestern University.}}
\author{%
\begin{tabular}{ccc}
\textsc{Elliot Lipnowski\footnote{e.lipnowski@columbia.edu}} & \textsc{Laurent Mathevet\footnote{laurent.mathevet@eui.eu}} & \textsc{Dong Wei\footnote{dwei10@ucsc.edu}}\\
{\small Columbia University} & {\small European University Institute} & {\small UC Santa Cruz}%
\end{tabular}
}
\maketitle
\begin{abstract}
\begin{changemargin}{1cm}{1cm} \begin{spacing}{1}
\noindent {\sc Abstract}. 
A well-intentioned principal provides information to a rationally inattentive agent without internalizing the agent's cost of processing information.  Whatever information the principal makes available, the agent may choose to ignore some. We study optimal information provision in a tractable model with quadratic payoffs where full disclosure is not optimal.
We characterize incentive-compatible information policies, that is, those to which the agent willingly pays full attention. In a leading example with three states, optimal disclosure involves information distortion at intermediate costs of attention.  As the cost increases, optimal information changes from downplaying the state to exaggerating the state.\vspace{0.2cm}\vspace{0.2cm}
\end{spacing}
\end{changemargin}
\end{abstract}

\newpage 

\begin{spacing}{1.1}
	\section{Introduction}	
	{\large Information is intended to benefit its recipient by allowing better decision making. Meanwhile, learning or information processing can be costly. These costs can cause the recipient to ignore some available information. 
	This motivated \cite{Simon1971,Simon1996} to call for the design of ``intelligent information-filtering systems.''
	
	In our recent paper \citep[``LMW'' hereafter]{Lipnowski2020}, we propose a theoretical framework to study the information-filtering problem aimed at managing a receiver's attention. Despite strong preference alignment, we show that a benevolent sender can have incentives to withhold information in order to induce an agent to pay ``better'' attention and make better decisions. While a necessary and sufficient condition is provided for \emph{when} information should be withheld, optimization (i.e., \emph{how} to optimally withhold information) is left open. The goal of this paper is to analyze this question in a tractable specification of LMW's model.
	
	We propose a quadratic model of a principal providing information to a rationally inattentive agent. The agent cares about his material loss, defined as the squared error between his action and the state, and also about the cost of processing information, defined as the squared Euclidean distance between his prior and posterior beliefs. 
	The principal, however, is motivated only by the agent's material welfare, as a teacher is motivated by her student's educational outcomes or a doctor by the fitness of her patient's medical decisions. Given such costs, the agent decides what information to acquire, taking whatever information the principal provides as an upper bound. Choosing said upper bound optimally is our design problem. 
	
	The key message from LMW is that attention management is fundamentally about trading off issues, that is, different dimensions of information. With two states, information can never be misused, because its sole use is to separate one state from the other; as a result, it is always optimal for the principal to provide full information. With more than two states, however, information becomes multidimensional. Intuitively, the agent needs to decide which aspects of the state to learn about, while at the same time determining the total amount of information to acquire. The principal may then strictly prefer to withhold information to restrict the agent's flexibility in garbling---though her sole objective is to help the agent make good decisions---so as to focus the agent on the most relevant issues.

	The quadratic model in the present paper is arguably the simplest available model with more than two states, because the agent's preferences over information are constant in a given ``direction.'' For this reason, the model can isolate the new forces that arise from multi-issue attention management. To study optimal disclosure, we first characterize (in Proposition \ref{prop:IC}) the information policies to which the agent willingly pays full attention. We find that incentive compatibility of distinguishing any given pair of messages in a policy can be determined by simple comparison of its marginal cost and benefit, and that pairwise incentive compatibility between messages implies that the information policy is incentive compatible.

\begin{spacing}{1.1}
	We apply our characterization of agent incentives to solve for the optimal attention management policy in a canonical example with three states. We explicitly derive the optimal policy for arbitrary prior beliefs and cost parameters, showing that it always takes a friendly form: The principal either downplays or exaggerates the state (Proposition \ref{prop:optimal}). In this example, we can interpret the principal as a teacher who hopes to maximize her student's performance on an (externally administered) exam in which the answer $\theta$ to each question is \quizT{} (1), \quizF{} (-1) or \quizU{} (0). The student can provide an explanation for his response, so that his action $a$ can be any number between -1 and 1, and his score, e.g., $4-(a-\theta)^2$, depends on how close his answer is to the correct one. At intermediate attention costs, partial disclosure is optimal even though the teacher wants the student to learn as much as possible. Specifically, when attention is relatively cheap, the teacher should present some of those questions with extreme answers (i.e., \quizT{} or \quizF{}) as uncertain (downplaying); when attention is relatively costly, the teacher should only emphasize the positive or negative side of those \quizU{} questions (exaggerating). By carefully withholding information, the teacher incentivizes the student to pay more attention to the available and more relevant aspects, leading to even better decision making.
	\bigskip
\end{spacing}
}

{\large 
\noindent\textbf{Related Literature.}
Our paper lies at the interface of two literatures: persuasion through flexible information \citep{Kamenica2011} and rational inattention \citep{Sims1998,Sims2003}.
The most related works, featuring well-intentioned information transmission with a flexibly inattentive audience, are those of \cite{Lipnowski2020} and \cite{Bloedel2018}.\footnote{In addition, \cite{Lester2012} analyze evidence exclusion in courts of law, where a judge chooses which of the finite pieces of evidence should be considered by the jury who then choose a subset of those to examine at a cost. They provide examples in which evidence exclusion leads to fewer sentencing errors. We study the same basic tradeoff in a flexible-information framework. 
	\cite{Wei2018} applies our framework to misaligned preferences, studying a buyer-seller setting.} 

In \cite{Lipnowski2020}, we propose the theoretical framework used in this paper, and find a necessary and sufficient condition for \emph{when} attention management is useful. The present paper studies \emph{how} to optimally manage attention, by specializing the framework to a tractable payoff environment. 

\begin{spacing}{1.1}
Our exercise is closest in spirit to that of \cite{Bloedel2018}, which studies a similar problem with substantially different modeling assumptions. Most notably, the qualitatively different cost specifications make the respective analyses relevant to distinct applications.\footnote{In addition to aligned preferences, \cite{Bloedel2018} also study the case of misaligned preferences, which lend the principal two reasons to withhold information: to bias the agent's decision, and to manipulate his attention.} When information is provided by another party (the principal), the agent's uncertainty could be about the realized message or about the underlying state of the world. In the model of \cite{Bloedel2018}, the agent bears an (entropy-based) cost to learn what message the principal has sent. In our model, following our previous paper, the agent bears a (quadratic) cost to learn the state. Thus, their cost specification is a cost of deciphering direct communication, while ours is better-suited to capture the cost of processing available data.
\end{spacing}
}
	
\section{The Model}\label{sec:2}
	{\large 
		We study the exact setting of LMW, specialized to a quadratic-payoff specification. Let $\Theta\subseteq \mathbb{R}$ be a finite set of states and $A$ be its convex hull. An agent must make a decision $a\in A$ in a world with uncertain state $\theta \in \Theta$ distributed according to a full-support prior distribution $\mu \in \Delta \Theta$.  When the agent chooses $a$ in state $\theta$, his material payoff is given by $u(a, \theta):=-(a-\theta)^2$. The principal's payoff is equal to the agent's material utility, $u$.
		
		In addition to his material utility, the agent also incurs an attention cost. As in the rational inattention literature, this cost is interpreted as the utility loss from processing information. To define it, first let
		$$
		\mathcal{R}\left(\mu\right):=\left\{p\in\Delta\Delta\Theta: \int_{\Delta\Theta}\nu\ddd p(\nu)=\mu\right\}
		$$
		be the set of \textbf{(information) policies}, which are the distributions over the agent's beliefs such that the mean equals the prior. It is well-known that signal structures and information policies are equivalent formalisms \citep[e.g.,][]{Kamenica2011}. 
		For a given cost parameter $\kappa>0$, we define the attention cost $C: \mathcal{R}\left(\mu\right) \rightarrow \mathbb{R}_+$ given by
		\begin{equation*}
		C(p)=\int_{\DT} c(\nu) \ddd p(\nu)=\int_{\DT}\kappa||\nu-\mu||^2 \ddd p(\nu),
		\end{equation*}
		where $||\cdot||$ is the Euclidean norm on $\mathbb{R}^{\Theta}$. 
		This cost function is posterior separable, that is, linear in the induced distribution of posterior beliefs.\footnote{For discussion and decision-theoretic foundations of posterior-separable costs, see \cite{caplin2017rationally, Morris2019,Hebert2019,Denti2020}.} 
		By Jensen's inequality, processing more information, in the sense of obtaining a policy $p$ more (Blackwell) informative than $q$, denoted $p \succeq^B q$,\footnote{\label{mps}For any $p,q\in \mathcal{R}(\mu)$, $p\succeq^B q$ if $p$ is a mean-preserving spread of $q$, that is, there is some measurable $r:\Delta\Theta\rightarrow\Delta\Delta\Theta$ such that (i) $p(S)=\int_{\Delta\Theta}^{}r(S|\cdot)\ddd q,\forall\text{ Borel }S\subseteq\Delta\Theta$ and (ii) $r(\cdot|\nu)\in\mathcal{R}(\nu),\forall \nu\in \Delta\Theta$.
		As is standard, this formalism is equivalent to one in which the principal chooses an experiment $\Theta\to\Delta M_P$, the agent (having seen the choice of experiment but not a realized message) chooses a Blackwell garbling $M_P\to\Delta M_A$, and the agent observes the realization (in $M_A$) of the composed experiment. The agent's garbling captures inattention and the loss of information it causes.} will incur a higher cost. 
		
				\medskip
		
		The timing of the game is as follows:
		\begin{itemize} 
	\item[--] The principal first chooses an information policy $p\in\RP(\mu)$. 
	\item[--] The agent then decides to what extent he should pay attention to $p$: He chooses a policy $q\in\RP(\mu)$ such that $q \preceq^B p$. Such a policy $q$ is called an {\bf (attention) outcome}.
	\item[--] Nature draws an agent belief $\nu\in\DT$ via $q$.
	\item[--] The agent chooses an action $a\in A$. 
	\item[--] Nature chooses a state $\theta\in\Theta$ via $\nu$.
\end{itemize}
		We study principal-preferred subgame perfect equilibria of this game.

        Define the principal's and the agent's indirect utilities at $\nu\in \Delta \Theta$ as:
        \begin{align*}
        U_P(\nu) &= U(\nu) := \max_{a\in A} \sum_{\theta\in \Theta}u(a,\theta)\nu(\theta),\\
        U_A(\nu) &= U(\nu)-c(\nu).
        \end{align*}
		The principal's problem can then be formalized as follows:
		\begin{equation}\label{eqn:designer}
		\begin{aligned}
		\sup_{p,q} & \displaystyle{\int_{\Delta\Theta} U_P \ddd q} \\
		&\text{s.t.}\enspace p \in\mathcal{R}\left(\mu\right) \enspace \mbox{and}\enspace q\in G^*(p),
		\end{aligned}
		\end{equation}
		where 
		$$G^*(p):= \argmax_{q\in\RP(\mu): \ q\preceq^B p} \left\{ \int_{\Delta\Theta} U \ddd q - C(q) \right\} = \argmax_{q\in\RP(\mu): \ q\preceq^B p} \int_{\Delta\Theta} U_A \ddd q$$
		is the agent's optimal garbling correspondence.	An information policy $p^*\in \mathcal{R}(\mu)$ is \textbf{(principal-) optimal} if $(p^*,q^*)$ solves \ref{eqn:designer} for some outcome $q^*\in\Delta\Delta\Theta$. The corresponding $q^*$ is an \textbf{optimal (attention) outcome}.
		
		As formalized, it is clear that the principal's problem is one of delegation. The policy $p$ chosen by the principal only appears in the constraint and does not directly affect any party's payoff. In effect, the principal makes available a menu of information policies, from which the agent picks his preferred one.
		
		We carry out our analysis under specific assumptions on the form of the agent's material payoff and his information processing cost. The foremost reason for our choice of the functional form is tractability: With quadratic preferences, we can fully characterize the set of incentive-compatible information policies, which serves as a foundation for finding an optimal policy.\footnote{Such tractability also enables complete solutions in \cite{ely2015suspense}, which uses the quadratic distance between beliefs to model preferences for suspense and surprise.} In addition to technical convenience, we see an important conceptual advantage to studying quadratic preferences. In short, quadratic preferences isolate the new forces that arise due to multidimensionality of the agent's problem. See Section \ref{sec:genquad} for a discussion on this point, and also for a more general class of quadratic preferences that enjoy the same advantages.

}

	\section{Simplifying Disclosure}\label{sec:3}
{\large In this section, we reduce the principal's disclosure problem to an amenable form. We first recall the existence result established in our previous paper, with an additional simplification, to reduce the constraint set to incentive-compatible policies with small support. We then provide a meaningful characterization of such policies, exploiting the quadratic structure of the present model. As such, we set the stage for deriving principles and solutions to the attention management problem.
\subsection{Rewriting the Principal's Program}

	Our previous paper establishes the existence of solution to the principal's problem, additionally showing that some optimum takes a special and convenient form.\footnote{The result we report here is stronger than Lemma 1 as stated in LMW. However, the proof for this strengthened result was recorded in LMW's Online Appendix.} Say that an information policy $p\in \mathcal{R}(\mu)$ is \textbf{incentive compatible (IC)} if the agent finds it optimal to pay full attention to it, that is, if $p\in G^*(p)$; and say it is {\bf nonredundant} if \mbox{\supp}$(p)$ is affinely independent.
	\begin{lemma}\label{thm:finite}
		There exists a nonredundant solution $q^*$ to 
		\begin{equation}\label{eqn:designer1}
		\begin{aligned}
		\max_q & \displaystyle{\int_{\DT} U_P \ddd p} \\
		&\text{s.t.} \hspace{0.5cm} (i) \hspace{0.4cm} q \in \mathcal{R}(\mu) \\
		& \hspace{1cm} (ii) \hspace{0.4cm}
		q \,\,\mbox{is IC}
		\end{aligned} 
		\end{equation}
		Moreover, $q^*$ solves \ref{eqn:designer1} if and only if $\left(p^*, q^*\right)$ solves \ref{eqn:designer} for some $p^*$.
	\end{lemma}

\begin{spacing}{1.1}
	The lemma shows that focusing on nonredundant, incentive-compatible information policies is without loss. 
	Since $\Theta$ is finite, nonredundancy implies the need for fewer messages than there are states.  Thus, Lemma \ref{thm:finite} reduces an infinite-dimensional optimization problem to a finite-dimensional one. In addition to having a small support, a nonredundant information policy enjoys another technical convenience: Its set of garblings is straightforward to describe, which simplifies the task of checking whether it is IC.  Specifically, for all $p,q\in \mathcal{R}(\mu)$ with $p$ nonredundant,\footnote{See \citet[Theorem 5]{wu2017coordinated}. He assumes $q$ to have finite support, but the general argument is identical.} 
	\begin{align*}
	p\succeq^{B}q \enspace \iff \enspace \supp(q) \subseteq \mbox{co}\left[\supp(p)\right].
	\end{align*}
\end{spacing}
}
	\subsection{Payoffs and Incentives}
	{\large
		
		Below, we index statistical measures, such as expectation $\E$ and variance $\V$, by the distribution of the underlying random variable.  For example, $\E_\nu\theta = \int_{\Theta} \tilde\theta\ddd\nu(\tilde\theta)$, $\V_\nu\theta = \E_\nu[\theta^2] - [\E_\nu \theta]^2$ and $\E_{p}\V_\nu\theta = \int_{\Delta \Theta}\V_\nu \theta\, \dd p(\nu)$.

				\bigskip
		
		\noindent {\bf Principal's Payoff}. The marginal distribution of actions is sufficient to compute our principal's expected payoffs, due to her simple ``match the state'' motive. 
		Indeed, as the agent's optimal action at any belief $\nu\in\DT$ is 
		$$
		a^*(\nu):=\argmax_{a\in A}u(a,\nu)=\E_{\nu}\theta,
		$$
		the principal's value is given by
		\begin{equation*}
		U_P(\nu) :=  - \E_{\nu} [ (\theta-a^*(\nu))^2 ]= -\V_{\nu}\theta,
		\end{equation*}
		which is strictly convex in $\nu$. 
		
		Take any incentive-compatible information policy $p\in \mathcal{R}(\mu)$, and note that 
		\begin{equation*}
		\int_{\DT}U_P \ddd p=
		-\E_{p}\V_{\nu}\theta
		=\V_{p}\E_{\nu}\theta-\V_{\mu}\theta
		=\V_{p}[a^*(\nu)]-\V_{\mu}\theta,
		\end{equation*}
		by the law of total variance. Therefore:
		\begin{observation}\label{obs:action}
			For any IC policies $p$ and $p'$, the principal strictly prefers $p$ to $p'$ if and only if $\V_{p}[a^*(\nu)]>\V_{p'}[a^*(\nu)]$.
		\end{observation}\medskip
		
		\noindent {\bf Psychological vs. Material Incentives}. 
		We now turn to the agent's incentives. Our goal is to characterize the set of IC policies which serves as the constraint in the principal's optimization problem. It turns out that IC can be characterized by a sequence of local comparisons between psychological cost and material benefit.
		
		The agent's attention cost at $\nu$ is 
		$$
		c(\nu) := \kappa ||\nu - \mu||^2 = \kappa \sum_\theta (\nu_\theta - \mu_{\theta})^2,
		$$
		so that his net indirect utility is 
		\begin{equation*}
		U_A(\nu) := U_P(\nu) - c(\nu) = \left[\left( \sum_{\theta} \nu_\theta \theta \right)^2 - \kappa \sum_{\theta} \nu_\theta^2\right] 
		+ {h(\nu)}
		\end{equation*}
		where $\nu_{\theta}:=\mathbb P(\theta|\nu)$ and $h(\nu)$ is affine in $\nu$. 
		
		To understand the agent's attentional tradeoff in this model, consider first a binary-support policy. Fix some $\nu,\nu^\prime$, and let $x_\theta:=\nu_\theta-\nu^\prime_{\theta}$ for all $\theta\in\Theta$. For $\epsilon\in [0,1]$, 
		\begin{eqnarray}\label{eqn:binaryIC}
		\frac{1}{2}\frac{\mbox{d}^2}{\mbox{d} \epsilon^2}U_A(\nu+\epsilon(\nu^\prime-\nu))
		&=& \left(\sum_{\theta} x_{\theta} \theta \right)^2 - \kappa \sum_{\theta} x_{\theta}^2 \\
		&=& | \E_{\nu} \theta - \E_{\nu'} \theta |^2 - \kappa || \nu - \nu'||^2\nonumber.
		\end{eqnarray}
		This derivative measures the curvature of $U_A$ between $\nu$ and $\nu'$ and brings about two notions of distance between beliefs:  
		\begin{itemize}
			\item[--] The \textbf{choice distance}, $| \E_{\nu} \theta - \E_{\nu'} \theta |$, which describes the change in action caused by a change in beliefs.
			\item[--] The \textbf{psychological distance}, $\sqrt{\kappa} || \nu - \nu'||$, which is proportional to a standard distance between beliefs.
		\end{itemize}
		
		The psychological distance measures the {marginal cost of extra attention} in a given direction, while the choice distance measures {the marginal instrumental benefit}. A policy with binary support $\{\nu,\nu'\}$ is IC if and only if $U_A$ is convex between $\nu$ and $\nu'$, that is, if and only if the choice distance exceeds the psychological distance.\footnote{This observation follows from Jensen's inequality, together with the observation that the univariate function $\epsilon\mapsto U_A(\nu+\epsilon(\nu^\prime-\nu))$ is (being quadratic) either globally weakly convex or globally strictly concave.}
		The next proposition generalizes this fact to all nonredundant policies.
		
		\begin{proposition}\label{prop:IC}
			A nonredundant policy $p\in \mathcal{R}(\mu)$ is IC if and only if 
			\begin{equation}\label{eq:orderIC}
			|a^{n+1}- a^{n} |\geq \sqrt{\kappa} || \nu^{n+1} - \nu^{n}|| \quad \forall n\in \{1,\ldots,N-1\},
			\end{equation} 
			where $\mbox{supp}(p) = \{\nu^1,\ldots,\nu^{N}\}$ such that $a^n:=\E_{\nu^n}(\theta)$ and $a^1\leq\cdots\leq a^N$.
		\end{proposition} 
		
		Condition \ref{eq:orderIC}, which we term {\bf order-IC}, is apparently much weaker than IC.  It can be interpreted as immunity to a particular small family of deviations by the agent: For each pair of consecutive (in terms of their induced action) messages, the agent would rather pay full attention than engage in the simple garbling that perfectly pools this pair and pays full attention to all other messages. In the present quadratic model, one such deviation is profitable if any deviation from full attention is profitable.\footnote{In the appendix, we formulate a more general quadratic model of attention management, allowing for multidimensional action and state spaces. Appropriately generalizing the definitions of choice and psychological distances, we can show that Proposition \ref{prop:IC} exactly goes through whenever the action space is unidimensional; moreover, even if the actions are not linearly ordered, an analogue of Proposition \ref{prop:IC} (without a reference to ``consecutive messages'') still holds.} 
		
		Beyond being interpretable, the characterization of IC in Proposition \ref{prop:IC} is useful. Below, we illustrate how it helps reduce the problem of optimal attention management. The next section will then use said reduction to fully solve a canonical example.
	
		\subsection{Implications of IC for Optimal Disclosure}\label{sec:3.3}
		For a nonredundant policy $p$ with support $\{\nu^n\}_{n=1}^N$ such that $N\geq 2$ and $\E_{\nu^1}(\theta)\leq \cdots \leq \E_{\nu^N}(\theta)$, let the \textbf{direction} of $p$ be the vector\footnote{Referring to ``the'' direction is a mild abuse of terminology when two distinct beliefs might induce the same action, leading to an ambiguous order over actions.  Because such information policies will obviously never be IC (as pooling those beliefs is strictly profitable for the agent), this ambiguity is immaterial for our purposes.}
		\begin{equation*}
		\left(\frac{\nu^{n+1} - \nu^n}{||\nu^{n+1} - \nu^n||}\right)_{n=1}^{N-1}. 
		\end{equation*}
		The direction of an information policy is the vector of normalized changes in uncertainty from one belief to the next. We say that a policy $p$ is \textbf{interior} if there exists another policy $p'$ with the same direction such that $p'\succ^B p$. The next corollary follows naturally from Proposition \ref{prop:IC} and Observation \ref{obs:action}.
		\begin{corollary}\label{cor:IC}
			A nonredundant policy can be an optimal attention outcome only if no interior and IC policy induces the same distribution of actions.
		\end{corollary}
		
		Corollary \ref{cor:IC} is helpful for ruling out a variety of information policies and reducing the dimension of our optimization problem. First, Corollary \ref{cor:IC} implies that an IC information policy admits a strict improvement if it is itself interior. Said differently, information should be maximal given its direction.  To illustrate graphically, let us consider the belief simplex in a three-state environment. Note that the direction of an information policy is unaffected by any affine transformation of beliefs (as in Figures \ref{fig:scaleup0} and \ref{fig:scaleup}) or changes in non-consecutive slopes (as in Figure \ref{fig:scaleup2}). Consequently, such transformations leave the comparisons in \ref{eq:orderIC} (hence, IC) unchanged, and each dashed policy in Figure \ref{fig:informationsaturation} is a strict improvement of the corresponding solid interior policy.
	
		\begin{figure}[ht!]
			\centering
			\begin{tabular}{ccc}  
				\begin{subfigure}[t]{.3\linewidth}
					\begin{tikzpicture}[scale=1.1]
					\draw (-2,0) -- (2,0) -- (0,2)--(-2,0);
					\node [below] at (0.3,0.7) {\small $\mu$};
					\draw[fill] (0.3,0.7) circle [radius=0.05];
					\draw[fill,blue] (-1,1) circle [radius=0.05];
					\draw[fill,blue] (1.6,0.4) circle [radius=0.05];
					\draw[fill,green] (-0.48,0.88) circle [radius=0.05];
					\draw[fill,green] (1.08,0.52) circle [radius=0.05];
					\draw [dashed, blue] (-1,1) --   (1.6,0.4);
					\draw [green] (-0.48,0.88) --   (1.08,0.52);
					\end{tikzpicture}
					\caption{}\label{fig:scaleup0}
				\end{subfigure}&
				\begin{subfigure}[t]{.3\linewidth}
					\begin{tikzpicture}[scale=1.1]
					\draw (-2,0) -- (2,0) -- (0,2)--(-2,0);
					\node [left] at (0.3,0.7) {\small $\mu$};
					\draw[fill] (0.3,0.7) circle [radius=0.05];
					\draw [red] (-1,0.5) -- (1,0.5) -- (0,1.2)-- (-1,0.5);
					\draw[fill,red] (-1,0.5) circle [radius=0.05];
					\draw[fill,red] (0,1.2) circle [radius=0.05];
					\draw[fill,red] (1,0.5) circle [radius=0.05];
					\draw [dashed,blue] (-1.8,0.2) -- (1.8,0.2) -- (0,1.46)-- (-1.8,0.2);
					\draw[fill,blue] (-1.8,0.2) circle [radius=0.05];
					\draw[fill,blue] (0,1.46) circle [radius=0.05];
					\draw[fill,blue] (1.8,0.2) circle [radius=0.05];				
					\end{tikzpicture}
					\caption{}\label{fig:scaleup}
				\end{subfigure}&
				\begin{subfigure}[t]{.3\linewidth}
					\centering
					\begin{tikzpicture}[scale=1.1]
					\draw (-2,0) -- (2,0) -- (0,2)--(-2,0);
					\node [above] at (0.3,0.7) {$\mu$};
					\draw[fill] (0.3,0.7) circle [radius=0.05];
					\draw [red] (2,0) -- (0,2) -- (-1,0.5)-- (2,0);
					\draw[fill,red] (2,0) circle [radius=0.05];
					\draw[fill,red] (0,2) circle [radius=0.05];
					\draw[fill,red] (-1,0.5) circle [radius=0.05];
					\draw [dashed,blue] (2,0) -- (0,2) -- (-1.25,0.125)-- (2,0);
					\draw[fill,blue] (-1.25,0.125) circle [radius=0.05];
					\end{tikzpicture}
					\caption{}\label{fig:scaleup2}
				\end{subfigure}
			\end{tabular}
			\caption{Excluding Interior Policies}\label{fig:informationsaturation}
		\end{figure}
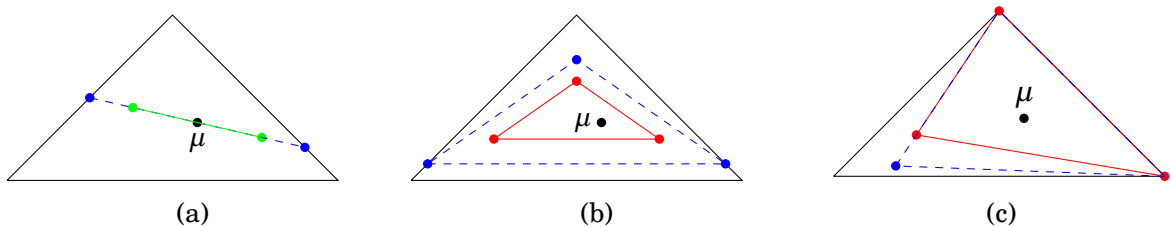
\begin{spacing}{1.1}
	Moreover, Corollary \ref{cor:IC} allows us to rule out some non-interior policies through a perturbation argument, by finding another interior policy that is IC and induces the same action distribution. To illustrate graphically, let us again consider the belief simplex in a three-state environment. In Figure \ref{fig:perturb}, the solid red policy is not interior, so it cannot be ruled out using the previous argument. Nonetheless, if we can find an interior perturbation (e.g., the dashed blue policy) to which it is easier to pay attention, while inducing the same action distribution (hence the same principal payoff),\footnote{Figure \ref{fig:perturb} represents beliefs by the probability $z$ that $\theta=0$ and the induced action $a$. This representation is also used in the next section.} then the dashed policy enables a strict improvement, making the original policy suboptimal.  
\end{spacing}	
	
	\begin{figure}[ht!]
		\centering
		\begin{tikzpicture}[scale=1.4]
		\draw (-2,0) -- (2,0) -- (0,2)--(-2,0);
		\node [below right] at (-0.03,0.8) {\small $\mu$};
		\node [above] at (2.5,0) {$a$};
		\node [right] at (0,2.1) {$z$};
		\draw[>=stealth, dashed,->]  (0,0) -- (0,2.2);
		\draw[>=stealth, dashed,->]  (-2.7,0) -- (2.7,0);
		\draw[dashed,blue] (-1.9,0.1) --  (0.7,1.1) -- (0.2,0.2) -- (-1.9,0.1);
		\draw[fill] (0.3,0.7) circle [radius=0.05];
		\draw [fill, red] (0.7,1.3) circle [radius=0.05];
		\draw [fill, red] (-1.9,0.1) circle [radius=0.05];
		\draw [fill, red] (0.2,0) circle [radius=0.05];
		\draw [fill, blue] (0.7,1.1) circle [radius=0.05];
		\draw [fill, blue] (0.2,0.2) circle [radius=0.05];
		\draw [red]  (0.7,1.3) -- (-1.9,0.1) -- (0.2,0) -- (0.7,1.3);
		\end{tikzpicture}
		\caption{Excluding Some Non-Interior Policies}\label{fig:perturb}
	\end{figure}
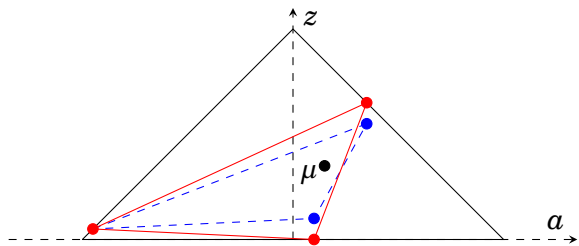	
	}
	\section{Optimal Attention Management: Three States}\label{sec:3state}
		{\large			
		We now apply our results in the previous section to pare down the search for optimal information in a model with three (evenly spaced) states, $\Theta = \{-1,0,1\}$.  We can interpret this scenario as a teacher creating a lesson plan to maximize her student's score in an exam where the answer to each question is \quizT{} (1), \quizF{} (-1) or \quizU{} (0).			
		
		We represent beliefs in the $(a,z)$ space, where $\z$ is the agent's belief that $\theta$ is zero and $a$ is his optimal action (see Figure \ref{fig:abc} in the Appendix).\footnote{The associated belief is
			$
			\nu_{(a,\z)}=\z\delta_0+\frac{1-\z+a}{2}\delta_1+\frac{1-\z-a}{2}\delta_{-1},
			$
			for $\z\in[0,1]$ and $a\in[\z-1,1-\z]$.}
		Rearranging \ref{eq:orderIC}, a policy $p$ with supported beliefs $\nu^n = \nu_{(a^n,\z^n)}$ is IC if and only if $\kappa \leq 2$ and
		\begin{equation*}
		\left|\frac{\z^{n+1}-\z^n}{a^{n+1}-a^n}\right| \leq s^*(\kappa):=\sqrt{\frac{2-\kappa}{3\kappa}}  \quad \forall n.\label{eqn:lineIC3states}
		\end{equation*}
		That is, a policy is IC if and only if the absolute slopes between consecutive beliefs are less than a cutoff.

		Before stating our characterization, we introduce a language to speak about some special policies. Say the principal {\bf downplays} the state if, for some $\pi_1,\pi_{-1}\in [0,1]$, she sends message:
		$$m=\begin{cases}
		0 &: \text{ with probability }1, \text{ if }\theta= 0\\
		\theta&:  \text{ with probability }1 - \pi_\theta, \text{ if }\theta\neq 0\\
		0&:  \text{ with probability }\pi_\theta, \text{ if }\theta\neq 0.
		\end{cases}$$
		The principal downplays the state if she sometimes says $0$ when the true state is one of the extremes. Complete downplaying ($\pi_{-1}=\pi_1=1$) conveys no information; no downplaying ($\pi_{-1}=\pi_1=0$) fully discloses the state; other forms convey partial information. Let $\left\{ p^D_{(\pi_{-1},\pi_1)}:\ (\pi_{-1},\pi_1) \in [0, 1]^2 \right\}$ denote the policies induced by downplaying the state. Say that the principal {\bf exaggerates} the state if, for some $\pi\in [0,1]$, she sends message:
		$$m=\begin{cases}
		\hspace{0.3cm}\theta&: \text{ with probability }1, \text{ if }\theta\neq 0\\
		\hspace{0.3cm}1&:  \text{ with probability }1-\pi, \text{ if }\theta= 0\\
		-1&:  \text{ with probability }\pi, \text{ if }\theta= 0.
		\end{cases}$$		
		The principal exaggerates the state if she reports an extreme state when the true state is 0. Increasing $\pi$ makes the agent more (less) certain that $\theta = m$ upon receiving $m=1$ ($m=-1$). Let {\bf separating exaggeration} refer to $\pi \in \{0, 1\}$. Policies induced by exaggerating the state are denoted $\left\{ p^E_\pi:\ \pi \in [0, 1] \right\}$.
		
		Let $a_\mu=\mu_1-\mu_{-1}$ be the agent's optimal action at the prior, and define
		\begin{eqnarray*}
		\kappa_1 & :=& \tfrac{1}{2} \nonumber\\ 
		\kappa_2 & := &\tfrac{2}{\frac{3}{4}\left(\frac{1-|a_\mu|+\mu_0}{1-|a_\mu|}\right)^2+1} \nonumber\\
		\kappa_3 &:= & \tfrac{2}{3\left(\frac{\mu_0}{1-|a_\mu|}\right)^2+1} \nonumber \\
		\kappa_4 & :=& 2.
		\end{eqnarray*}
		Note that $\kappa_1<\kappa_2<\kappa_3<\kappa_4$ for any interior $\mu$.
		
		\begin{proposition}\label{prop:optimal}
			In the three-state model, an optimal attention outcome:
			\begin{enumerate}
				\item fully reveals the state when $\kappa\in \left(0,\kappa_1\right]$;\label{optimal1}
				\item downplays the state until \ref{eq:orderIC} holds with equality when $\kappa\in \left(\kappa_1,\kappa_2\right]$;\label{optimal2}
				\item engages in separating exaggeration when $\kappa\in \left[\kappa_2,\kappa_3\right]$;\label{optimal3}\footnote{In particular, $\pi=1$ ($\pi=0$) is optimal if and only if $a_\mu\geq 0$ ($a_\mu\leq 0$).}
				\item exaggerates the state until \ref{eq:orderIC} holds with equality when $\kappa\in \left(\kappa_3,\kappa_4\right]$;\label{optimal4}	
				\item reveals no information when $\kappa\in \left(\kappa_4,\infty\right)$.\label{optimal5}\vspace{-0.2cm}
			\end{enumerate}
			Moreover, for generic $\kappa$, every optimal attention outcome is of the above form.\footnote{When $\kappa=\kappa_2$, both downplaying and separating exaggeration are optimal, and continuum many optima exist that mix these two policies.  For any other $\kappa$, there is a unique optimal attention outcome (up to reversing the role of $1$ and $-1$ in the case that the prior is symmetric).}
		\end{proposition}
		
				\begin{figure}[ht!]
			\centering
			\begin{tikzpicture}[scale=0.8]
			\draw (-2,0) -- (2,0) -- (0,2)--(-2,0);
			\node [above] at (0.3,0.7) {$\mu$};
			\node [below] at (2.8,0) {$a$};
			\node [right] at (0,2.3) {$z$};
			\node [below] at (0,-0.2) {$\kappa\leq\kappa_1$};
			\node [below] at (0,-0.7) {{\footnotesize(Full revelation)}};
			\draw[>=stealth, dashed,->]  (0,0) -- (0,2.5);
			\draw[>=stealth, dashed,->]  (-2.7,0) -- (3,0);
			\draw[fill] (0.3,0.7) circle [radius=0.05];
			\draw[fill,red] (0,2) circle [radius=0.05];
			\draw[fill,red] (-2,0) circle [radius=0.05];
			\draw[fill,red] (2,0) circle [radius=0.05];
			\draw [red] (-2,0) -- (0,2) -- (2,0)-- (-2,0);
			\end{tikzpicture}
			\begin{tikzpicture}[scale=0.8]
			\draw (-2,0) -- (2,0) -- (0,2)--(-2,0);
			\node [above] at (0.3,0.6) {$\mu$};
			\node [below] at (2.8,0) {$a$};
			\node [right] at (0,2.3) {$z$};
			\node [below] at (0,-0.2) {$\kappa_1<\kappa\leq \kappa_2$};
			\node [below] at (0,-0.7) {{\footnotesize(Downplaying)}};
			\draw[>=stealth, dashed,->]  (0,0) -- (0,2.5);
			\draw[>=stealth, dashed,->]  (-2.7,0) -- (3,0);
			\draw[fill] (0.3,0.7) circle [radius=0.05];
			\draw[fill,red] (2,0) circle [radius=0.05];
			\draw[fill,red] (-2,0) circle [radius=0.05];
			\draw[fill,red] (0,1.5) circle [radius=0.05];
			\draw [red] (-2,0) -- (0,1.5) -- (2,0)-- (-2,0);
			\draw (0,2)--(0.2,2);
			\draw (0,1.5) -- (0.2,1.5);
			\draw [decorate,decoration={brace,amplitude=3pt,mirror},xshift=0pt,yshift=0pt]
			(0.2,1.5) -- (0.2,2);
			\node [right] at (0.2,1.75) {\footnotesize frequency of mistakes};
			\end{tikzpicture}\\
			\begin{tikzpicture}[scale=0.8]
			\draw (-2,0) -- (2,0) -- (0,2)--(-2,0);
			\node [above] at (0.3,0.7) {$\mu$};
			\node [below] at (2.8,0) {$a$};
			\node [right] at (0,2.3) {$z$};
			\node [below] at (0,-0.2) {$\kappa_2\leq \kappa\leq \kappa_3$};
			\node [below] at (0,-0.7) {{\footnotesize (Separating exaggeration)}};
			\draw[>=stealth, dashed,->]  (0,0) -- (0,2.5);
			\draw[>=stealth, dashed,->]  (-2.7,0) -- (3,0);
			\draw[fill] (0.3,0.7) circle [radius=0.05];
			\draw[fill,green] (2,0) circle [radius=0.05];
			\draw[fill,green] (-0.83,1.17) circle [radius=0.05];
			\draw [green] (2,0) --   (-0.83,1.17);
			\end{tikzpicture}
			\begin{tikzpicture}[scale=0.8]
			\draw (-2,0) -- (2,0) -- (0,2)--(-2,0);
			\node [above] at (0.3,0.7) {$\mu$};
			\node [below] at (2.8,0) {$a$};
			\node [right] at (0,2.3) {$z$};
			\node [below] at (0,-0.2) {$\kappa_3\leq \kappa\leq \kappa_4$};
			\node [below] at (0,-0.7) {{\footnotesize (Exaggeration)}};
			\draw[>=stealth, dashed,->]  (0,0) -- (0,2.5);
			\draw[>=stealth, dashed,->]  (-2.7,0) -- (3,0);
			\draw[fill] (0.3,0.7) circle [radius=0.05];
			\draw[fill,green] (-1,1) circle [radius=0.05];
			\draw[fill,green] (1.6,0.4) circle [radius=0.05];
			\draw [green] (-1,1) --   (1.6,0.4);
			\end{tikzpicture}
			\begin{tikzpicture}[scale=0.8]
			\draw (-2,0) -- (2,0) -- (0,2)--(-2,0);
			\node [above] at (0.3,0.7) {$\mu$};
			\node [below] at (2.8,0) {$a$};
			\node [right] at (0,2.3) {$z$};
			\node [below] at (0,-0.2) {$\kappa>\kappa_4$};
			\node [below] at (0,-0.7) {{\footnotesize(No disclosure)}};
			\draw[>=stealth, dashed,->]  (0,0) -- (0,2.5);
			\draw[>=stealth, dashed,->]  (-2.7,0) -- (3,0);
			\draw[fill,red] (0.3,0.7) circle [radius=0.05];
			\end{tikzpicture}
			\caption{Optimal attention outcomes}
			\label{fig:optimal}
		\end{figure}
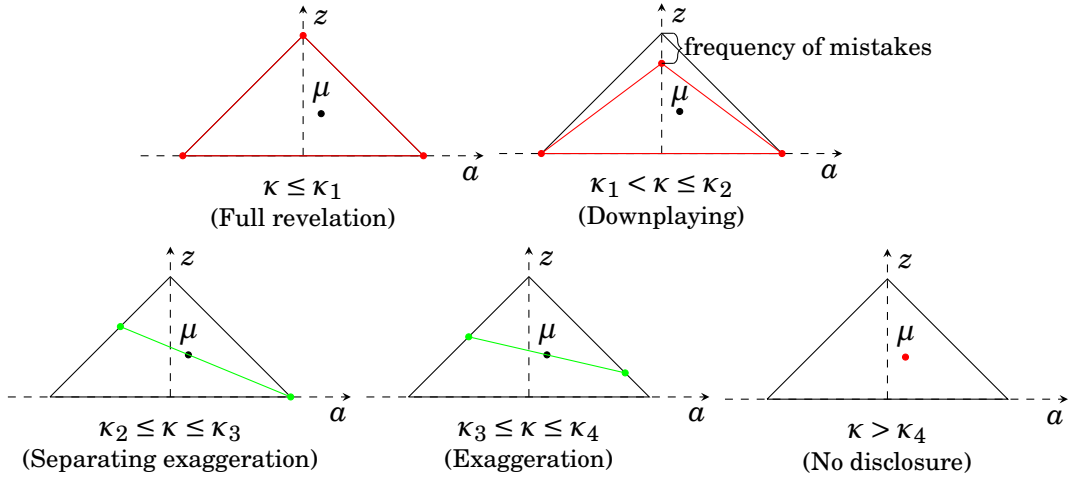

		When attention comes at low cost, the principal keeps quiet (i.e., reports 0) on some extreme occasions and is truthful the rest of the time. That is, she downplays the state. As shown in Figure \ref{fig:optimal}, this invites either inaction from the agent, who chooses 0, or an extreme reaction $\{-1, 1\}$. The latter are never mistaken, because they happen in precisely the principal's reported extreme state. However, inaction is harmful in extreme situations, so the agent makes {all} his mistakes through action 0. The frequency of mistakes, shown in Figure \ref{fig:optimal}, is chosen so that the agent is barely willing to pay full attention. Under this type of information disclosure, mistakes are large, always of size 1, but they are kept {rare} as long as $\kappa \leq \kappa_2$.
		
		As inattention grows more severe, the principal increasingly downplays the state to keep the agent's attention. Eventually, the principal switches to exaggerating the state to avoid the potentially harmful consequences of extreme behaviors and maintain the agent's full attention. Exaggeration results in smaller mistakes (than under downplaying, because their size is strictly less than 1), although they occur {more} often.

Our result implies that for intermediate attention costs, the principal can improve the material welfare of the agent by excluding some  information. This observation recalls principles from the delegation literature \citep{Szalay2005} showing that restricting an agent's choices can incentivize information acquisition. By downplaying or exaggerating the state, the principal endogenously restricts the agent's behavior, activating the same incentive channel.

}

	\section{Conclusion}\label{sec:5}
{\large	\begin{spacing}{1.1}
		This paper aims at understanding how a benevolent principal should optimally withhold information to help an inattentive agent make informed decisions. Building on the framework of \cite{Lipnowski2020}, 
		we present a tractable environment with quadratic payoffs where incentive-compatible information can be fully characterized. We illustrate the implications of incentive compatibility on optimal disclosure, and further apply them to derive optimal information policies in a canonical example with multidimensional information. As such, this paper provides insight into how to optimally manage attention.
	\end{spacing}
}

\bibliographystyle{jpe}
\bibliography{attentionm}

\begin{thebibliography}{16}
\newcommand{\enquote}[1]{``#1''}
\providecommand{\natexlab}[1]{#1}
\providecommand{\url}[1]{\texttt{#1}}
\providecommand{\urlprefix}{URL }

\bibitem[{Bloedel and Segal(2021)}]{Bloedel2018}
Bloedel, Alexander and Ilya Segal. 2021.
\newblock \enquote{Persuading a Rationally Inattentive Agent.}
\newblock Working Paper.

\bibitem[{Caplin, Dean, and Leahy(2021)}]{caplin2017rationally}
Caplin, Andrew, Mark Dean, and John Leahy. 2021.
\newblock \enquote{Rationally inattentive behavior: Characterizing and
  generalizing Shannon entropy.}
\newblock Working Paper.

\bibitem[{Denti(2021)}]{Denti2020}
Denti, Tommaso. 2021.
\newblock \enquote{Posterior Separable Cost of Information.}
\newblock Working Paper.

\bibitem[{Ely, Frankel, and Kamenica(2015)}]{ely2015suspense}
Ely, Jeffrey, Alexander Frankel, and Emir Kamenica. 2015.
\newblock \enquote{{Suspense and Surprise}.}
\newblock \emph{Journal of Political Economy} 123~(1):215--260.

\bibitem[{H\'{e}bert and Woodford(2021)}]{Hebert2019}
H\'{e}bert, Benjamin and Michael Woodford. 2021.
\newblock \enquote{Rational Inattention when Decisions Take Time.}
\newblock Working Paper.

\bibitem[{Kamenica and Gentzkow(2011)}]{Kamenica2011}
Kamenica, Emir and Matthew Gentzkow. 2011.
\newblock \enquote{Bayesian Persuasion.}
\newblock \emph{American Economic Review} 101~(6):2590--2615.

\bibitem[{Lester, Persico, and Visschers(2012)}]{Lester2012}
Lester, Benjamin, Nicola Persico, and Ludo Visschers. 2012.
\newblock \enquote{{Information Acquisition and the Exclusion of Evidence in
  Trials}.}
\newblock \emph{Journal of Law, Economics, and Organization} 28~(1):163--182.

\bibitem[{Lipnowski, Mathevet, and Wei(2020)}]{Lipnowski2020}
Lipnowski, Elliot, Laurent Mathevet, and Dong Wei. 2020.
\newblock \enquote{Attention Management.}
\newblock \emph{American Economic Review: Insights} 2~(1):17--32.

\bibitem[{Morris and Strack(2019)}]{Morris2019}
Morris, Stephen and Philipp Strack. 2019.
\newblock \enquote{The Wald Problem and the Relation of Sequential Sampling and
  Ex-Ante Information Costs.}
\newblock Working Paper.

\bibitem[{Simon(1971)}]{Simon1971}
Simon, Herbert~A. 1971.
\newblock \enquote{Designing Organizations for an Information Rich World.}
\newblock In \emph{Computers, Communications, and the Public Interest}, edited
  by Martin Greenberger. Baltimore, 37--72.

\bibitem[{Simon(1996)}]{Simon1996}
---{}---{}---. 1996.
\newblock \emph{The Sciences of the Artificial}.
\newblock Cambridge, MA, USA: MIT Press.

\bibitem[{Sims(1998)}]{Sims1998}
Sims, Christopher. 1998.
\newblock \enquote{Stickiness.}
\newblock \emph{Carnegie-Rochester Conference Series on Public Policy}
  49~(1):317--356.

\bibitem[{Sims(2003)}]{Sims2003}
---{}---{}---. 2003.
\newblock \enquote{Implications of Rational Inattention.}
\newblock \emph{Journal of Monetary Economics} 50~(3):665--690.

\bibitem[{Szalay(2005)}]{Szalay2005}
Szalay, Dezs\"o. 2005.
\newblock \enquote{The Economics of Clear Advice and Extreme Options.}
\newblock \emph{Review of Economic Studies} 72~(4):1173--1198.

\bibitem[{Wei(2021)}]{Wei2018}
Wei, Dong. 2021.
\newblock \enquote{Persuasion Under Costly Learning.}
\newblock \emph{Journal of Mathematical Economics} 94:102451.

\bibitem[{Wu(2018)}]{wu2017coordinated}
Wu, Wenhao. 2018.
\newblock \enquote{Sequential Bayesian Persuasion.}
\newblock Working paper.

\end{thebibliography}
	\appendix
	\section{Appendix}
	\subsection{A General Quadratic Model of Attention Management}\label{sec:genquad}
	In this section, we outline a model of attention management with general quadratic preferences. Such a model is of interest for the following reason. The main result of LMW is that full information is universally optimal (that is, optimal for every choice problem and preference specification) if and only if the state space is binary. One can reinterpret this result as saying that any scope for beneficial attention management derives from the agent's space of beliefs being multidimensional. We see quadratic preferences as particularly well-suited to isolate the richness imparted by multidimensional beliefs, as the benefits and costs of information---captured by the degree of convexity of $U_A$ and $c$---are constant along every ``one-dimensional slice'' of the space of beliefs.\footnote{More precisely, the second derivative in condition \ref{eqn:binaryIC} is independent of $\epsilon$.}
		
	The generalized model is exactly as in Section \ref{sec:2}, except with a more general specification of the state space, action space, material objective, and cost of information. Given $J\in\mathbb N$, let $\Theta\subseteq\real^J$ be some finite nonempty set, and $A$ be the convex hull of $\Theta$ (or some compact subset of $\real^J$ that contains it). Given arbitrary vectors $\gamma,\pi\in\real^\Theta$ and arbitrary positive semidefinite matrices $\Gamma\in\real^{J\times J}$ and $\Pi\in\real^{\Theta\times\Theta}$, assume the material objective $u:A\times\Theta\to\real$ is given by $u(a,\theta):=\gamma_{\theta}-(a-\theta)^\top \Gamma (a-\theta)$, and the attention cost functional $c:\DT\to\real$ is given by $c(\nu):=\nu^\top \pi + (\nu-\mu)^\top \Pi (\nu-\mu)$, where $\mu\in\DT$ is the prior. This specification captures all finite-state, quadratic models in which the agent optimally chooses her expectation $a^*(\nu)=\sum_{\theta\in\Theta} \nu_\theta\theta$ of the state (and moreover, faces a Blackwell-monotone cost of information). The formulation of Section \ref{sec:2} corresponds to the case in which $J=1$, $\gamma=\pi=0$, $\Gamma=1$, and $\Pi=\kappa I_{\Theta}$ for some $\kappa>0$, where $I_\Theta$ is the identity matrix.
	
	Simple computations show that 
	\begin{equation*}
		U_A(\nu) = U_P(\nu) - c(\nu) =  a^*(\nu)^\top \Gamma a^*(\nu)- \nu^\top \Pi\nu
		+ {h(\nu)}
	\end{equation*}
	for some affine $h:\DT\to\real$. 
	Define now the seminorms $||\cdot||_M$ on $\real^J$ via $||y||_M:=\sqrt{y^\top \Gamma y}$ and $||\cdot||_P$ on $\real^\Theta$ via $||x||_M:=\sqrt{x^\top \Pi x}$. Given two beliefs $\nu,\nu'\in\DT$, we can define a \textbf{choice distance} between them as $|| \E_{\nu'}\theta - \E_{\nu}\theta||_M = ||a^*(\nu')-a^*(\nu)||_M$, and a \textbf{psychological distance} as $||\nu'-\nu||_P$; both of these formulae specialize to the forms given in our main model. Further replicating the computations before Proposition \ref{prop:IC}, for $\epsilon\in [0,1]$, 
	\begin{eqnarray*}
		\frac{1}{2}\frac{\mbox{d}^2}{\mbox{d} \epsilon^2}U_A(\nu+\epsilon(\nu^\prime-\nu))
		&=& || \E_{\nu'} \theta - \E_{\nu} \theta ||_M^2 - || \nu' - \nu||_P^2.
	\end{eqnarray*}
	Hence, the following generalization of Proposition~\ref{prop:IC} can be established.\footnote{Moreover, if $J=1$ and $\Gamma$ is nonzero, then Corollary~\ref{cor:IC} applies, with an identical proof. To generalize to the multidimensional-action case, one requires a more comprehensive notion of ``the direction of an information policy'' that involves every pair of supported posteriors, not only consecutive ones. Adopting an analogous, more stringent definition of an interior information policy---which in general makes the result less powerful by ruling out a smaller family of information policies---an analogue Corollary~\ref{cor:IC} holds with an identical argument. The analysis of Section~\ref{sec:3state} is, of course, specific to the fully parametric environment of that section. }
	
	\begin{customthm}{1*}\label{prop:IC-general}
		In the generalized quadratic setting, a nonredundant policy $p\in \mathcal{R}(\mu)$ is IC if and only if 
		\begin{equation}\label{eq:orderIC-general}
			|| \E_{\nu'} \theta - \E_{\nu} \theta ||_M \geq || \nu' - \nu||_P \quad \forall \nu,\nu'\in\supp(p).
		\end{equation} 
		Moreover, if $J=1$, and $\supp(p)=\{\nu^1,\ldots,\nu^{N}\}$ is such that $a^n:=\E_{\nu^n}(\theta)$ satisfy $a^1\leq\cdots\leq a^N$, then it suffices to verify \ref{eq:orderIC-general} only for the case that $\nu=\nu^n$ and $\nu'=\nu^{n+1}$ for some $n\in\{1,\ldots,N-1\}$.
	\end{customthm}

	\subsection{Proof of Proposition \ref{prop:IC}}

	\begin{lemma}\label{lem:convex}
		In the general quadratic model of Section~\ref{sec:genquad}, for any $\nu,\nu'\in\Delta\Theta$, the function $U_A|_{\textrm{co}\{\nu,\nu^\prime\}}$ is convex [resp. strictly concave] if and only if $|| \E_{\nu'} \theta - \E_{\nu} \theta ||_M \geq [<] || \nu' - \nu||_P$. In particular, in the model of Section~\ref{sec:2}, this ranking reduces to $| \E_{\nu'} \theta - \E_{\nu} \theta | \geq [<] \sqrt{\kappa}|| \nu' - \nu||$.
	\end{lemma}
	\begin{proof}\hspace{.1cm}
Given any $\nu\in\DT$, the agent's first-order condition is clearly solved at $a=a^*(\nu):=\sum_{\theta\in\Theta}\nu_\theta \theta$, and the objective is concave in the chosen action because $\Gamma$ is positive semidefinite. Hence $a^*(\nu)$ is an optimal action. Therefore, 
\begin{eqnarray*}
U_A(\nu) &=& U_P(\nu) - c(\nu) \\
&=& \int \left\{ \gamma(\theta)-[a^*(\nu)-\theta]^\top \Gamma [a^*(\nu)-\theta] - \pi(\theta)\right\} \ddd\nu(\theta) -  \nu^\top \Pi \nu \\
&=& \int \left[ \gamma(\theta)-\theta^\top \Gamma \theta - \pi(\theta)\right] \ddd\nu(\theta) +  a^*(\nu)^\top \Gamma a^*(\nu) - \nu^\top \Pi \nu.
\end{eqnarray*}
Now, for $\nu,\nu'\in\Delta\Theta$ and $\epsilon\in[0,1]$, we have
		\begin{eqnarray*}
		\tfrac{1}{2}\tfrac{\dd^2}{\dd \epsilon^2}U_A(\nu+\epsilon(\nu^\prime-\nu))
		&=& \tfrac{1}{2}\tfrac{\dd^2}{\dd \epsilon^2}  \left\{ a^*\left(\nu+\epsilon(\nu'-\nu)\right)^\top \Gamma a^*\left(\nu+\epsilon(\nu'-\nu)\right) - \left[\nu + \epsilon(\nu'-\nu) \right]^\top \Pi \left[\nu + \epsilon(\nu'-\nu) \right] \right\}\\
		&=&  \left[ a^*(\nu')-a^*(\nu) \right]^\top \Gamma \left[ a^*(\nu')-a^*(\nu) \right] - (\nu'-\nu)^\top \Pi (\nu'-\nu)\\
		&=& || a^*(\nu')-a^*(\nu) ||_M^2 - || \nu' - \nu||_P^2.
		\end{eqnarray*}
The general result follows. Finally, to specialize the result to our main model, note that $|| \cdot ||_M = |\cdot|$ and $||\cdot||_P=\sqrt{\kappa}||\cdot||$ in this case.
	\end{proof}
	\medskip

Let us now prove Proposition~\ref{prop:IC-general}, from which Proposition~\ref{prop:IC} follows directly.

	\hspace{-1.25cm}\begin{proof}[Proof of Proposition \ref{prop:IC-general}]
		Fixing any nonredundant $p\in \mathcal{R}(\mu)$, we will first show that the following are equivalent:
		\begin{enumerate}
			\item $p$ is IC.
			\item $U_A|_{\textrm{co}\{\nu^\prime,\nu^{\prime\prime}\}}$ is weakly convex for all $\forall \nu',\nu'' \in \rm{supp}(p)$.
			\item $|| \E_{\nu''} \theta - \E_{\nu'} \theta ||_M \geq || \nu'' - \nu'||_P$ for all $\nu',\nu''\in\supp(p)$.
		\end{enumerate}
		That {(2) $\iff$ (3)} follows directly from Lemma \ref{lem:convex}, so we now turn to showing {(1) $\iff$ (2)}.
		
		{(1) $\Rightarrow$ (2)}: Suppose condition (2) does not hold.  As $U_A$ is quadratic, there then exist $\nu^\prime,\nu^{\prime\prime}\in \supp(p)$ such that $U_A|_{\textrm{co}\{\nu^\prime,\nu^{\prime\prime}\}}$ is strictly concave. Define the finite-support random posterior $q\in G(p)$ by $$q(\nu)=\begin{cases}
		0 &: \ \nu \in \{\nu',\nu''\}\\
		p(\nu^\prime)+p(\nu^{\prime\prime}) &: \ \nu=\tfrac{p(\nu^\prime)}{p(\nu^\prime)+p(\nu^{\prime\prime})}\nu^\prime+\tfrac{p(\nu^{\prime\prime})}{p(\nu^\prime)+p(\nu^{\prime\prime})}\nu^{\prime\prime} \\
		p(\nu) &: \ \text{otherwise.}
		\end{cases}$$		
		In words, $q$ is a random posterior that replaces $\nu^\prime,\nu^{\prime\prime}$ in $\supp(p)$ with their conditional mean. By construction, $p\succ^B q$, so that $q\in G(p)$. Also, since $U_A|_{\textrm{co}\{\nu^\prime,\nu^{\prime\prime}\}}$ is strictly concave, by Jensen's inequality,
		$$\int_{\Delta\Theta}^{}U_A\ddd q> \int_{\Delta\Theta}^{}U_A\ddd p.$$
		This implies that $p\notin G^*(p)$, i.e. (1) does not hold.\\
		
		{(2) $\Rightarrow$ (1)}: Suppose (2) holds. By Claims OA.4 and OA.6 of \cite{Lipnowski2020}, there exists some Blackwell-maximal element $q^*\in G^*(p)$. Since $\Theta$ is finite and $S:=\supp(p)$ is affinely independent, we know that $S$ is finite, and that the map $\beta:\Delta S\rightarrow\mbox{co}(S)$, defined by $p\mapsto \sum_{\nu\in S}^{}p(\nu)\nu$, is bijective.
		
		Assume (toward a contradiction) that $\exists \nu^*\in \supp(q^*)\setminus S$. As $\beta$ is bijective and $\nu^*\notin S$, we know that $\beta^{-1}(\nu^*)$ is not a point mass. So $\exists \nu^\prime,\nu^{\prime\prime}\in S$ and $\varepsilon>0$ such that $\beta^{-1}(\nu^\prime|\nu^*),\beta^{-1}(\nu^{\prime\prime}|\nu^*)>2\varepsilon$. By continuity, there exists a neighborhood $N\subseteq\mbox{co}(S)$ of $\nu^*$ such that:
		$$\beta^{-1}(\nu^\prime|\nu),\beta^{-1}(\nu^{\prime\prime}|\nu)>\varepsilon,\forall \nu\in N.$$
		Note that $q^*(N)>0$ as $\nu^*\in \supp(q^*)$. 
		
		Define $f_+, f_-:\Delta\Theta\rightarrow\Delta\Theta$ by $\nu\mapsto \nu+\varepsilon(\nu^\prime-\nu^{\prime\prime})\mathbf1_{\nu\in N}$ and $\nu\mapsto \nu-\varepsilon(\nu^\prime-\nu^{\prime\prime})\mathbf1_{\nu\in N}$, respectively.\footnote{Notice that since $\nu^*\notin S$, one can choose $\varepsilon$ and $N$ small enough to ensure $f_+$ and $f_-$ are well-defined $\DT$-valued maps.} Also, define
		$$q^\prime:=\tfrac{1}{2}q^*\circ f_+^{-1}+\tfrac{1}{2} q^*\circ f_-^{-1}\in \Delta\Delta\Theta.$$
		By construction, $q^\prime\succ^B q^*$. Moreover, since $U_A|_{\textrm{co}\{\nu^\prime,\nu^{\prime\prime}\}}$ is convex, Lemma \ref{lem:convex} tells us that $U_A|_{\textrm{co}\{\nu\pm\varepsilon(\nu^\prime-\nu^{\prime\prime})\}}$ is also convex, $\forall \nu\in N$. This implies
		$$\int_{\Delta\Theta}^{}U_A\ddd q^\prime\geq \int_{\Delta\Theta}^{}U_A\ddd q^*,$$
		so that $q^\prime\in G^*(p)$, contradicting maximality of $q^*$. Therefore, $\supp(q^*)\subseteq S$, i.e. $q^*\in \Delta S$. But $\beta$ is bijective and $\beta(q^*)=\mu=\beta(p)$, so that $q^*=p$. Hence, $p\in G^*(p)$, i.e. (1) holds.\\
		
Finally, having proved the three-way equivalence, we now specialize the model to the case of $J=1$. Because $||\cdot||_M$ is a seminorm on $\real$, it follows that $||\cdot||_M=\lambda|\cdot|$ for some $\lambda\in\real_+$. 
Let $N:=|\supp(p)|$, and write $\supp(p)=\{\nu^1,\ldots,\nu^N\}$, where $a^n:=a^*(\nu^n)$ satisfy $a^1\leq\cdots\leq a^N$. Suppose $||a^{n+1}-a^n||_M\geq ||\nu^{n+1}-\nu^n||_P$ for every $n\in\{1,\ldots,N-1\}$.

Fixing some $i,j\in\{1,\ldots,N\}$ with $i<j$, we now argue that $||a^j-a^i||_M\geq ||\nu^j-\nu^i||_P$. Consider two cases. First, suppose $a^j=a^i$. In this case, every $n\in\{i,\ldots, j-1\}$ has $a^i\leq a^n \leq a^{n-1} \leq a^j$, so that $||a^{n+1}-a^n||_M=0$, and hence $||\nu^{n+1}-\nu^n||_P=0$. Because $\nu^j-\nu^i=\sum_{n=i}^{j-1} (\nu^n-\nu^{n-1})$, the triangle inequality yields $||\nu^j-\nu^i||_P\leq \sum_{n=i}^{j-1} ||\nu^n-\nu^{n-1}||_P$, giving the desired inequality. Now focus on the complementary case that $a^j>a^i$. Direct computation then shows: \begin{eqnarray*}
			\frac{ \nu^{j} - \nu^{i}}{a^{j} - a^{i}} &=& \tfrac{1}{\sum_{\ell=i}^{j-1} \left( a^{\ell+1} - a^{\ell} \right)} \sum_{n=i}^{j-1} \left( \nu^{n+1} - \nu^{n} \right) \\
			&=& \sum_{n=i}^{j-1} \left[ \tfrac{a^{n+1} -a^n}{\sum_{\ell=i}^{j-1} \left( a^{\ell+1} - a^{\ell} \right)} \left( \frac{\nu^{n+1} - \nu^{n}}{a^{n+1} -a^n}\right) \right] \\
			&\in& {\rm{co}} \left\{ \frac{\nu^{n+1} - \nu^{n}}{a^{n+1} -a^n} \right\}_{n=i}^{j-1}.
		\end{eqnarray*}
		As a seminorm is convex, and therefore quasiconvex, it follows that $$\left\| \frac{ \nu^{j} - \nu^{i}}{a^{j} - a^{i}} \right\|_P \leq \max_{\ell\in\{i,\ldots,j-1\}} \left\| \frac{ \nu^{\ell+1} - \nu^{\ell}}{a^{\ell+1} - a^{\ell}} \right\|_P \leq \lambda,$$
		and hence $\left\| \nu^{j} - \nu^{i} \right\|_P \leq \lambda|a^j-a^i|= \left\| a^{j} - a^{i} \right\|_M$, as required.
	\end{proof}
\medskip

\hspace{-1.25cm}\begin{proof}[Proof of Corollary \ref{cor:IC}]
	By Observation \ref{obs:action}, we only need to show that any (nonredundant) policy that is IC and interior cannot be an optimal attention outcome.\footnote{Note that the definition of an interior policy involves its direction, which is only defined for nonredundant policies and which is distinct by definition for any two policies with different support size (see Section \ref{sec:3.3}). Hence, by definition, an interior policy is both nonredundant and informative.} Let $p$ be a policy that is IC and interior. By definition, there exists an alternative nonredundant policy $\tilde{p}$ with the same direction as $p$ such that $\tilde{p}\succ^B p$. Since the principal strictly prefers the agent to process more information, we are done if $\tilde{p}$ is IC. But by Proposition \ref{prop:IC}, preserving the direction of an IC information policy preserves IC.
 \end{proof}

\section{Online Appendix: Proof of Proposition \ref{prop:optimal}}
In this Online Appendix, we provide the complete proof of Proposition \ref{prop:optimal}, which characterizes the optimal disclosure policy in the three-state model.

	Information policies with two supported messages are called \textbf{binary policies}, and those with three supported messages are called \textbf{ternary policies}. With three states, any nonredundant, informative policy is either binary or ternary.
	
	In what follows, we represent the belief space $\DT$ parametrically. Let $$\B=\left\{(a,\z)\in \mathbb{R}^2:\ \z\in [0,1],\ a\in [\z-1,1-\z]\right\}.$$ 
	For any $(a,z)\in \B$, let 
	\begin{align*}
	\nu_{(a,\z)}=\z\delta_0+\frac{1-\z+a}{2}\delta_1+\frac{1-\z-a}{2}\delta_{-1}.
	\end{align*}
	So $\nu_{(a,\z)}\in\DT$ has $a=\E_{\theta\sim \nu_{(a,\z)}}(\theta)$ and $z= \nu_{(a,\z)}(0)$. By construction, the map $(a,\z)\mapsto \nu_{(a,\z)}$ is an affine bijection from $\B$ to $\DT$.
	Under this representation, Dirac measures in $\Delta\Theta$ are extreme points of $\B$: $\delta_{-1}=\nu_{(-1,0)}$, $\delta_{0}=\nu_{(0,1)}$, and $\delta_{1}=\nu_{(1,0)}$. Also, the prior satisfies $\mu=\nu_{(a_\mu,\mu_0)}$, where $a_\mu=\mu_1-\mu_{-1}$ and $(a_\mu,\mu_0)\in \mbox{int}(\B)$. Figure \ref{fig:abc} depicts the $(a,z)$ representation of the belief space. 
	
	For any distinct beliefs $\nu=\nu_{(a,\z)}$ and $\nu'=\nu_{(a',\z')}$, let $(\Delta a,\Delta \z)=(a-a',\z-\z')$. Note that
	\begin{align}\label{eqn:edgeIC3states}
	&|a-a'|\geq\sqrt{\kappa}||\nu-\nu'||\nonumber\\
	\iff& |a-a'|\geq \sqrt{\kappa}\sqrt{(z-z')^2+\left(\frac{1-z+a}{2}-\frac{1-z'+a'}{2}\right)^2+\left(\frac{1-z-a}{2}-\frac{1-z'-a'}{2}\right)^2}\nonumber\\
	\iff& \kappa\leq 2,\ \Delta a \neq 0, \text{ and } \left| \frac{\Delta z}{\Delta a} \right| \leq \sqrt{\frac{2-\kappa}{3\kappa}} =: s^*(\kappa).
	\end{align}
	By Proposition \ref{prop:IC}, a nonredundant $p\in \mathcal{R}(\mu)$ is IC, if and only if \ref{eqn:edgeIC3states} holds for any two consecutive messages in $\supp(p)$.
	Finally, we fix notation for line segments that form the boundary of $\B$: $\lt:=\co\{(0,1),(-1,0)\}$, $\rt:=\co\{(0,1),(1,0)\}$ and $\down:=\co\{(-1,0),(1,0)\}$.
		\begin{figure}[ht!]
		\centering
		\begin{tikzpicture}
		\draw (-2,0) -- (2,0) -- (0,2)--(-2,0);
		\node [below] at (0.7,0.7) {$(a_\mu,\mu_0)$};
		\node [above] at (2.8,0) {$a$};
		\node [right] at (0,2.5) {$z$};
		\node [below] at (-2,0) {$(-1,0)$};
		\node [right] at (0,2) {$(0,1)$};
		\node [below] at (2,0) {$(1,0)$};
		\node [right] at (1,1.1) {$\rt$};
		\node [left] at (-1,1.1) {$\lt$};
		\node [below] at (0,-0.1) {$\down$};
		\draw[>=stealth, dashed,->]  (0,0) -- (0,2.7);
		\draw[>=stealth, dashed,->]  (-2.7,0) -- (3,0);
		\draw[fill] (0.3,0.7) circle [radius=0.05];
		\end{tikzpicture}
		\caption{$(a,z)$-representation of $\Delta\Theta$}\label{fig:abc}
	\end{figure}
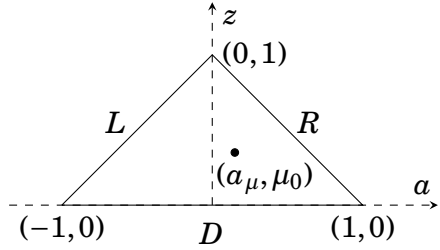

	\begin{claim}\label{cl11}
		Let $p\in \mathcal{R}(\mu)\setminus\{\delta_\mu\}$ be nonredundant and IC, with support $\{\nu_{(a^n,\z^n)}\}_{n=1}^m$ such that $a^1<\cdots< a^m$. If $(a^1,\z^1)\in \mbox{int}(\B)$ or $(a^m,\z^m)\in \mbox{int}(\B)$, then $p$ is interior.
	\end{claim}
	\begin{proof}\hspace{.1cm}
		Suppose that $(a^1,\z^1)\in\mbox{int}(\B)$.  For sufficiently small $\epsilon>0$, interiority of $(a^1,\z^1)$ implies that $$(\tilde a, \tilde \z):= (a^2,\z^2) + (1+\epsilon)(a^1-a^2,\z^1-\z^2)\in \B,$$ and $\{\nu_{(\tilde a, \tilde \z)}, \ \nu_{(a^2,\z^2)},\ldots,\nu_{(a^m,\z^m)}\}$ is still affinely independent.  There is then a mean-preserving spread $\tilde p$ of $p$ which is supported on $\{\nu_{(\tilde a, \tilde \z)}, \ \nu_{(a^2,\z^2)},\ldots,\nu_{(a^m,\z^m)}\}$.\footnote{To construct it, take a dilation on $\DT$ which fixes each of $\nu_{(a^2,\z^2)},...,\nu_{(a^m,\z^m)}$ and splits $\nu_{(a^1, \z^1)}$ to a measure with support $\{\nu_{(\tilde a, \tilde \z)}, \ \nu_{(a^2,\z^2)}\}$.}  As $\frac{\z^2-\tilde \z}{a^2-\tilde a} = \frac{\z^2-\z^1}{a^2-a^1}$ by construction, $\tilde{p}$ has the same direction as $p$. Since $\tilde{p}$ is a mean-preserving spread of $p$, by definition we have $\tilde{p}\succ^B p$, and so $p$ is interior. 
		
		A symmetric argument proves that $p$ is interior if $(a^m,\z^m)\in\mbox{int}(\B)$.
	\end{proof}
\bigskip
	
	Below, we rule out ternary policies such that (when parametrized in $\B$) the middle message lies below the line segment between the two other messages.
	
	\begin{claim}\label{cl12}
		Take any nonredundant ternary policy $p\in \mathcal{R}(\mu)$ with $\supp(p)=\left\{\nu_{(a_1,\z_1)},\nu_{(a_2,\z_2)},\nu_{(a_3,\z_3)}\right\}$ such that $a_1<a_2<a_3$. If $p$ is IC and $$\z_2\leq \tfrac{a_3-a_2}{a_3-a_1} \z_1+ \tfrac{a_2-a_1}{a_3-a_1} \z_3,$$ then $p$ is not principal-optimal.
	\end{claim}
	\begin{proof}\hspace{.1cm}
		First,	let us assume that $\z_1\leq \z_3$. Nonredundancy of $p$ rules out the possibility that $\z_2 = \tfrac{a_3-a_2}{a_3-a_1} \z_1+ \tfrac{a_2-a_1}{a_3-a_1} \z_3$; therefore, $$\z_2< \tfrac{a_3-a_2}{a_3-a_1} \z_1+ \tfrac{a_2-a_1}{a_3-a_1} \z_3\leq \z_3.$$ 
		Letting $s_\ell:= \frac{\z_2-\z_1}{a_2 - a_1}$ and $s_r:= \frac{\z_3-\z_2}{a_3 - a_2}$, observe that 
		$$ s_r-s_\ell = \tfrac{1}{a_3 - a_2} \z_3 + \tfrac{1}{a_2 - a_1}\z_1 - \left(\tfrac{1}{a_2 - a_1} + \tfrac{1}{a_3 - a_2} \right)\z_2 = \tfrac{a_3 - a_1}{(a_2 - a_1)(a_3 - a_2)}\left[ \tfrac{a_3-a_2}{a_3-a_1} \z_1+ \tfrac{a_2-a_1}{a_3-a_1} \z_3 - z_2 \right]>0.$$
		This implies $s_r>0$, because $$(a_2-a_1)s_\ell + (a_3-a_2)s_r  = \z_3-\z_1\geq 0.$$
		Letting $p_i:=p\left(\nu_{(a_i,\z_i)}\right)\in(0,1)$ for $i\in\{1,2,3\}$, consider a small $\epsilon>0$, and define
		\begin{align*}
		\left(\tilde{a}_1,\tilde{\z}_1\right)&=\left(a_1,\z_1\right)\\
		\left(\tilde{a}_2,\tilde{\z}_2\right)&=\left(a_2,\z_2+p_3\epsilon\right)\\
		\left(\tilde{a}_3,\tilde{\z}_3\right)&=\left(a_3,\z_3-p_2\epsilon\right).
		\end{align*}
		All statements in what follows should be taken to mean, \emph{when $\epsilon>0$ is sufficiently small}.

		Applying Jensen's inequality to the concave function $1-|\cdot|$ , $$\z_2< \tfrac{a_3-a_2}{a_3-a_1} \z_1+ \tfrac{a_2-a_1}{a_3-a_1} \z_3 \leq \tfrac{a_3-a_2}{a_3-a_1} (1-|a_1|)+ \tfrac{a_2-a_1}{a_3-a_1} (1-|a_3|) \leq 1-|a_2|.$$
		Combining this with $\z_3>z_2\geq 0$ tells us that $\left(\tilde{a}_2,\tilde{\z}_2\right),\left(\tilde{a}_3,\tilde{\z}_3\right)\in \mbox{int}\left(\B\right)$.  Next, define $\tilde p\in\Delta\DT$ to be the measure which puts mass $p_i$ on ${\nu}_{\left(\tilde{a}_i,\tilde{\z}_i\right)}$ for $i\in\{1,2,3\}$.  Direct computation shows that $\tilde p\in\RP(\mu)$ and that $\tilde p$ generates exactly the same action distribution as $p$.  
		Now, to show that $\tilde p$ is IC, observe that $\tfrac{\tilde{\z}_3-\tilde{\z}_2}{\tilde{a}_3-\tilde{a}_2} = \tfrac{\tilde{\z}_3-\tilde{\z}_2}{{a}_3-{a}_2} \in [0,s_r]$ and $\tfrac{\tilde{\z}_2-\tilde{\z}_1}{\tilde{a}_2-\tilde{a}_1} = \tfrac{\tilde{\z}_2-\tilde{\z}_1}{{a}_2-{a}_1} \in [s_\ell, s_r]$.   So $\tilde p$ is order-IC if $p$ is.  Appealing to Proposition \ref{prop:IC} says $\tilde p$ is IC. Finally, since $(\tilde{a}_3,\tilde{z}_3)\in \mbox{int}\left(\B\right)$, Claim \ref{cl11} tells us $\tilde{p}$ is interior. Hence, we have constructed an interior and incentive-compatible $\tilde{p}$ which generates the same action distribution as $p$. Thus, by Corollary \ref{cor:IC}, $p$ is not principal-optimal.
		
		A symmetric argument proves the same result in the case $\z_1\geq \z_3$.
	\end{proof}

Without loss of generality, \emph{we focus hereafter on the case that $a_\mu\geq 0$.}
	
	\begin{claim}\label{cl13}
		Suppose a nonredundant $p\in \mathcal{R}(\mu)$ is IC. If there exists $ \nu_{(a,z)}\in \supp(p)$ such that $(a,z)\in \mbox{relint}\left(\down\right)$, then $p$ is not principal-optimal.
	\end{claim}
	\begin{proof}\hspace{.1cm}
		Say $p$ has support $\{\nu_{(a_1,\z_1)},\ldots,\nu_{(a_m,\z_m)}\}$ such that $a_1<\cdots< a_m$ and $m\in\{1,2,3\}$.  
		First, let us assume $\z_1\leq z_m$.
		
		If $m=1$, then the result follows from $\mu_0>0$.  If $m=3$ and $\z_2\leq \tfrac{a_3-a_2}{a_3-a_1} \z_1+ \tfrac{a_2-a_1}{a_3-a_1} \z_3$, then the result follows from Claim \ref{cl12}. We focus on the case that $m\in\{2,3\}$; $z_1=0$; and if $m=3$, then $\z_2> \tfrac{a_2-a_1}{a_3-a_1} \z_3$.
		
		Letting $p_i:=p\left(\nu_{(a_i,\z_i)}\right)\in(0,1)$ for $i\in\{1,\ldots,m\}$, consider a small $\epsilon>0$ and, define, $$ \left(\tilde{a}_i,\tilde{\z}_i\right)=\begin{cases}
		\left(a_1,\z_1+p_2\epsilon \right) & \ i=1 \\
		\left(a_2,\z_2-p_1\epsilon \right) & \ i=2 \\
		\left(a_i,\z_i \right) & \ \text{otherwise}. \\
		\end{cases}$$
		All statements in what follows are taken to mean, \emph{when $\epsilon>0$ is sufficiently small}.  
		
		That $\z_2>0$ and $(a_1, z_1)\in \mbox{relint}\left(\down\right)$ imply $\left(\tilde{a}_1,\tilde{\z}_1\right),\left(\tilde{a}_2,\tilde{\z}_2\right)\in \mbox{int}\left(\B\right)$.  Defining $\tilde p\in\Delta\DT$ to be the measure which puts mass $p_i$ on ${\nu}_{\left(\tilde{a}_i,\tilde{\z}_i\right)}$ for $i=1,\ldots,m$, observe that $\tilde p\in\RP(\mu)$ and that $\tilde p$ generates exactly the same action distribution as $p$.  Just as in the proof of Claim \ref{cl12}, an appeal to Proposition \ref{prop:IC}, Claim \ref{cl11}, and Corollary \ref{cor:IC} means we need only show that $\tilde p$ is order-IC.
		
		But see that $\tfrac{\tilde\z_2-\tilde\z_1}{\tilde a_2- \tilde a_1} = \tfrac{\z_2-\z_1 - (p_1+p_2)\epsilon}{a_2- a_1} \in \left[ 0,\tfrac{\z_2-\z_1}{a_2-a_1}\right]$.  Moreover, in the ternary case, $\tfrac{\z_2-\z_1}{a_2-a_1} < \tfrac{\z_3-\z_2}{a_3-a_2}$ because $\z_2> \tfrac{a_2-a_1}{a_3-a_1} \z_3$, so that $\tfrac{\tilde\z_3-\tilde\z_2}{\tilde a_3- \tilde a_2} = \tfrac{\z_3-\z_2 + p_2\epsilon}{a_3- a_2} \in \left[\tfrac{\z_2-\z_1}{a_2-a_1},\ \tfrac{\z_3-\z_2}{a_3-a_2} \right]$. The result follows. 
		
		A symmetric argument proves the same result in the case $\z_1\geq \z_3$.
	\end{proof}

	\begin{lemma}\label{lem:optimalbinary}
		Take any nonredundant binary policy $p\in \mathcal{R}(\mu)$ with $\supp(p)=\left\{\nu_{(a_1,\z_1)},\nu_{(a_2,\z_2)}\right\}$ such that $a_1<a_2$. Suppose $(p,p)$ is a solution to \ref{eqn:designer} (i.e. $p$ is an optimal attention outcome). Then,
		\begin{enumerate}
			\item $(a_1,\z_1)\in \lt$, $(a_2,\z_2)\in \rt$;
			\item $\left|\frac{\z_2-\z_1}{a_2-a_1}\right|\leq s^*(\kappa)$
		\end{enumerate}
	\end{lemma}
	\begin{proof}\hspace{.1cm}
		Claims \ref{cl11} and \ref{cl13} imply that $(a_1,\z_1),(a_2,\z_2)\in \lt\cup\rt$; that $p\in\RP(\nu_{(a_\mu,\mu_0)})$ and $(a_\mu,\mu_0)\in\mbox{int}(\B)$ then imply part (1).
		Part (2) follows from Proposition \ref{prop:IC} and equation \ref{eqn:edgeIC3states}.
	\end{proof}
	
	\begin{lemma}\label{lem:optimalternary}
		Suppose $\frac{1}{2}<\kappa< 2$. Take any nonredundant ternary policy $p\in \mathcal{R}(\mu)$ with $\supp(p)=\left\{\nu_{(a_1,\z_1)},\nu_{(a_2,\z_2)},\nu_{(a_3,\z_3)}\right\}$ such that $a_1<a_2<a_3$. Suppose $(p,p)$ is a solution to \ref{eqn:designer} (i.e. $p$ is an optimal attention outcome). Then,
		\begin{enumerate}
			\item $(a_2,\z_2)\in \mbox{int}\left(\B\right)$, $(a_1,\z_1)\in \lt$, $(a_3,\z_3)\in \rt$;
			\item $\frac{\z_2-\z_1}{a_2-a_1}=s^*(\kappa)$, $\frac{\z_3-\z_2}{a_3-a_2}=-s^*(\kappa)$.
		\end{enumerate}
	\end{lemma}
	
	\hspace{-1.25cm}\begin{proof}\hspace{.1cm}
		Since $p$ is nonredundant and IC, Claims \ref{cl11} and \ref{cl13} imply that $(a_1,\z_1),(a_3,\z_3)\in \lt\cup\rt$. As $s^*(\kappa)<1$ when $\kappa>\tfrac12$, it cannot be (given Proposition \ref{prop:IC}) that $\lt$ contains two distinct beliefs in the support of $p$; and similarly for $\rt$. This implies that $(a_1,\z_1)\in \lt$, $(a_3,\z_3)\in \rt$, and $(a_2, \z_2)\in \B\setminus(\lt\cup\rt)$.  But Claim \ref{cl13} rules out $(a_2, \z_2)\in \down$ as well. This delivers part (1).

		Now we prove part (2).  To start, it follows from Proposition \ref{prop:IC} and equation \ref{eqn:edgeIC3states} that $\left| \frac{\z_2-\z_1}{a_2-a_1}\right| ,\  \left| \frac{\z_3-\z_2}{a_3-a_2}\right|\leq s^*(\kappa)$.

		We now establish that $\left|\frac{\z_2-\z_1}{a_2-a_1}\right|=\left|\frac{\z_3-\z_2}{a_3-a_2}\right|=s^*(\kappa)$. Assume (toward a contradiction) that this does not hold.  
		Let us first assume $\left|\frac{\z_2-\z_1}{a_2-a_1}\right|<s^*(\kappa)$. 
		Consider a small $\epsilon>0$, and define the three vectors,
		\begin{align*}
		\left(\tilde{a}_1,\tilde{\z}_1\right)&=\left(a_1,\z_1\right)\\
		\left(\tilde{a}_2,\tilde{\z}_2\right)&=(1+\epsilon)\left(a_2,\z_2\right) - \epsilon\left(a_3,\z_3\right)\\
		\left(\tilde{a}_3,\tilde{\z}_3\right)&=\left(a_3,\z_3\right).
		\end{align*}
		Taking $\epsilon>0$ to be sufficiently small, $\{\left(\tilde{a}_i,\tilde{\z}_i\right)\}_{i=1}^3$ will be three affinely independent vectors, all in $\B$, whose convex hull contains $(0,\mu_0)$; and $\left|\frac{\tilde\z_2-\tilde\z_1}{\tilde a_2-\tilde a_1}\right|<s^*(\kappa)$.  There is therefore a unique $\tilde p\in \RP(\mu)$ whose support is $\{\left(\tilde{a}_i,\tilde{\z}_i\right)\}_{i=1}^3$, and $\tilde p$ is order-IC, since $\left|\frac{\tilde\z_3-\tilde\z_2}{\tilde a_3-\tilde a_2}\right|=\left|\frac{\z_3-\z_2}{ a_3- a_2}\right|\leq s^*(\kappa)$.  Proposition \ref{prop:IC} guarantees that $\tilde p$ is IC.  But, by construction, the action distribution induced by $\tilde p$ is a strict mean-preserving spread of that induced by $p$.  Therefore, by Observation \ref{obs:action}, $p$ is not optimal, a contradiction.
		
		A symmetric argument derives the same contradiction in the case $\left|\frac{\z_3-\z_2}{a_3-a_2}\right|<s^*(\kappa)$. Therefore, we have $\left|\frac{\z_2-\z_1}{a_2-a_1}\right|=\left|\frac{\z_3-\z_2}{a_3-a_2}\right|=s^*(\kappa)$.

		Finally, nonredundancy implies that $\frac{\z_2-\z_1}{a_2-a_1}\neq\frac{\z_3-\z_2}{a_3-a_2}$, and Claim \ref{cl12} rules out the possibility that $\frac{\z_2-\z_1}{a_2-a_1}<0<\frac{\z_3-\z_2}{a_3-a_2}$.  Part (2) then follows.
	\end{proof}
	
	\bigskip
	
	We define three special information policies that will be used in the coming proofs.
	\begin{itemize}
		\item The ``full disclosure'' policy $p^F$ is such that $\supp(p^F)=\{\nu_{(-1,0)},\nu_{(0,1)},\nu_{(1,0)}\}$
		\item The ``no information'' policy $p^N$ is such that $\supp(p^N)=\{\mu\}$
		\item The orthogonal policy $p^O$ is such that $\supp(p^O)=\{\nu_{(\mu_0-1,\mu_0)},\nu_{(1-\mu_0,\mu_0)}\}$ 
	\end{itemize}

	The above lemmas can be combined into the following claim, which reduces our search for optimal attention outcomes to a single two-dimensional problem.
	
	\begin{claim}\label{cl14}
		Suppose $\tfrac12<\kappa< 2$, and suppose the nonredundant policy $p\in \mathcal{R}(\mu)$ is an optimal attention outcome.  Then there exist $(a_1,z_1), (a_2,z_2), (a_3,z_3)\in \B$ such that:
		\begin{itemize}
			\item $(a_1,\z_1)\in \lt$, $(a_3,\z_3)\in \rt$, and $a_1\leq a_2 \leq a_3$;
			\item $\z_2-\z_1=s^*(\kappa)(a_2-a_1)$, $\z_3-\z_2=-s^*(\kappa)(a_3-a_2)$;
			\item $p\{\nu_{(a_1,z_1)}, \nu_{(a_2,z_2)}, \nu_{(a_3,z_3)}\}=1$.
		\end{itemize}
	\end{claim}
	
	\hspace{-1.25cm}\begin{proof}\hspace{.1cm}
		As $\kappa\leq 2$, the policy $p^O$ is IC, generating strictly higher principal payoffs than no information.  So $p$ must be informative.  Being nonredundant, it is either ternary or binary.
		
		In the case that $p$ is ternary, Lemma \ref{lem:optimalternary} delivers the result.
		
		In the remaining case, $p$ has binary support $\left\{\nu_{(a_1,\z_1)},\nu_{(a_3,\z_3)}\right\}$, where $(a_1,z_1),(a_3,z_3)\in \B$ with $a_1<a_3$.  Lemma \ref{lem:optimalbinary} then tells us that $(a_1,\z_1)\in \lt$, $(a_3,\z_3)\in \rt$, and $\left|\frac{\z_3-\z_1}{a_3-a_1}\right|\leq s$, where $s:= s^*(\kappa)$.  Moreover, $s\in (0,1)$ because $\tfrac12<\kappa< 2$.
		
		Now, the line of slope $s$ through $(a_1,\z_1)$ and the line of slope $-s$ through $(a_3,\z_3)$ have different slopes (since $s>0$), so that there they have a unique intersection point $(a_2,\z_2)\in \real^2$.  To prove the claim, then, all that remains is to show is that $a_1\leq a_2\leq a_3$ and $(a_2,\z_2)\in \B$.  
		
		Define, for $\tilde a_2\in\real$, the gap $g(\tilde a_2) = [\z_3 - s(\tilde a_2-a_3)] - [\z_1 + s(\tilde a_2-a_1)]$, which strictly decreases in $\tilde a_2$, so that $g^{-1}(0)= \{a_2\}$.  Since $\left|\frac{\z_3-\z_1}{a_3-a_1}\right|\leq s$, $g(a_3)\leq0\leq g(a_1)$.  The intermediate value theorem then says that $a_2\in [a_1, a_3]$.
		
		As $s\leq 1$, $a_2\geq a_1$, and $(a_1,\z_1)$ lies (weakly) below the line containing $\lt$, it follows that $(a_2,\z_2)$ lies below that line as well. 
		As $s\leq 1$, $a_2\leq a_3$, and $(a_3,\z_3)$ lies (weakly) below the line containing $\rt$, it follows that $(a_2,\z_2)$ lies below that line as well. 
		As $s\geq 0$, $a_2\geq a_1$, and $(a_1,\z_1)$ lies (weakly) above the line containing $\down$, it follows that $(a_2,\z_2)$ lies above that line as well.  Therefore $(a_2,\z_2)\in \B$.
	\end{proof}
	
	\begin{lemma}\label{lem:binorsym}
		If $\tfrac12<\kappa< 2$, then there exists an optimal attention outcome $p$ such that one of the following holds
		\begin{itemize}
			\item $p$ is binary: $\supp(p)= \{\nu_{(a_1, \z_1)},\ \nu_{(a_2, \z_2)}\}$ for some $(a_1, \z_1)\in \lt, \ (a_2, \z_2)\in\rt$.
			\item $p$ is ternary with critical slopes and support containing $\nu_{(1,0)}$:\footnote{Recall our normalization that $a_\mu\geq 0$.} $\supp(p)= \{\nu_{(a_1, \z_1)},\ \nu_{(a_2, \z_3)}, \ \nu_{(1, 0)}\}$ for some $(a_1, \z_1)\in \lt,\ (a_2,z_2)\in \mbox{int}({\B})$ with $\frac{\z_2-\z_1}{a_2-a_1}=s=\frac{z_2}{1-a_2}$
		\end{itemize}
		Moreover, policies of the second form exist if and only if $s^*(\kappa)>\frac{\mu_0}{1-a_\mu}$.
	\end{lemma}
	
	\hspace{-1.25cm}\begin{proof}\hspace{.1cm}
		Here, $s:= s^*(\kappa)\in (0,1)$ because $\tfrac12<\kappa< 2$. 
		
		Define the set $\mathcal T\subseteq \real^6$ of tuples $((a_i,\z_i))_{i=1}^3\subseteq B^3$ such that:\begin{itemize}
			\item $(a_1,\z_1)\in \lt$, $(a_3,\z_3)\in \rt$, and $a_1\leq a_2 \leq a_3$;
			\item $\z_2-\z_1=s(a_2-a_1)$, $\z_3-\z_2=-s(a_3-a_2)$;
			\item $(a_\mu,\mu_0)\in\co\{(a_1,z_1),\ (a_2,z_2),\ (a_3,z_3)\}$.
		\end{itemize}
		As $s>0$, it follows that every element $t\in\mathcal T$ has $\{t_1,t_2,t_3\}$ an affinely independent subset of $\B$.\footnote{In the extreme case where $a_2=a_1$ (or $a_2=a_3$), $(a_2,z_2)$ and $(a_1,z_1)$ (or $(a_2,z_2)$ and $(a_3,z_3)$) collapse to one point, so that $\{t_1,t_3\}$ is still an affinely independent subset of $\B$.}  Given that $(a_\mu,\mu_0)\in \co \{t_1,t_2,t_3\}$, it follows that there is a unique $p^{\mathcal{T}}_t \in \RP(\mu)\cap \Delta\{\nu_{t_1},\nu_{t_2},\nu_{t_3}\}$.
		The policy $p^{\mathcal{T}}_t$ is IC by Proposition \ref{prop:IC}.  Moreover, Claim \ref{cl14} tells us that any nonredundant optimal attention outcome is of this form.  Accordingly, we can reformulate the principal's problem as $\max_{t\in \mathcal T} \int_\DT U_P \ddd p^{\mathcal{T}}_t$. 
		
		Toward a parametrization of $\mathcal T$, let $\mathcal A:= \{(a_1, a_2):\ t=(a_1, \z_1, a_2, \z_2, a_3, \z_3)\in \mathcal T\}.$
		
		Given $t=(a_1, \z_1, a_2, \z_2, a_3, \z_3)\in \mathcal T$, we can infer the following:  \begin{itemize}
			\item That $(a_1,\z_1)\in \lt$ implies $\z_1=1+ a_1$.
			\item Then, $\z_2 = \z_1 + s(a_2-a_1)= 1+ (1-s)a_1 + s a_2$.
			\item Finally, $(a_3, \z_3)$ belongs both to $\rt$ (so that $\z_3=1-a_3$) and to the line of slope $-s$ through $(a_2,\z_2)$---uniquely pinning it down as these two lines have different slopes.  Direct computation then shows that $(a_3,\z_3)=(-a_1-\tfrac{2s}{1-s}a_2, \ 1+a_1+\tfrac{2s}{1-s}a_2).$
		\end{itemize}

		So, for any $(a_1,a_2)\in \real^2$, let $$t(a_1,a_2):= (a_1,\ 1+ a_1,\ a_2, \ 1+ (1-s)a_1 + s a_2,\ -a_1-\tfrac{2s}{1-s}a_2, \ 1+a_1+\tfrac{2s}{1-s}a_2) \in \real^6. $$
		Consistent with the previous notation, let
		\begin{align*}
		t_1(a_1,a_2)&=(a_1,\ 1+ a_1),\\
		t_2(a_1,a_2)&=(a_2,1+ (1-s)a_1 + s a_2),\\
		t_3(a_1,a_2)&=(-a_1-\tfrac{2s}{1-s}a_2, \ 1+a_1+\tfrac{2s}{1-s}a_2).
		\end{align*} 
		The above derivations show that $\mathcal T=\{t(a_1,a_2)\}_{(a_1,a_2)\in \mathcal A}$. Hence, we can view the principal's problem as $\max_{(a_1,a_2)\in \mathcal A} \int_\DT U_P \ddd p^{\mathcal{T}}_{t(a_1,a_2)}$.
		
		Now we consider two cases separately.
		
		\begin{itemize}
			\item If $(a_1,a_2)\in \mathcal A$ and $t_2(a_1,a_2)\neq t_3(a_1,a_2)$ (i.e. $a_2\neq -\frac{1-s}{1+s}a_1$), then 
			$$a_1\left(1-s\right)+a_2\left(1+s\right)\neq 0.$$

			Moreover, $a_1<a_2$.\footnote{This is because, if $a_1=a_2$, then $(a_1,\z_1)=(a_2,\z_2)=(0,1)$. But then, the fact that $(a_3,\z_3)\in \rt$ implies that $(a_\mu,\mu_0)\in \rt$, contradicting $(a_\mu,\mu_0)\in \mbox{int}(\B)$.}  Therefore, given $(a_1,a_2)\in \mathcal A$, the following three numbers are all well-defined: \begin{eqnarray*}
				p_2(a_1,a_2)&:=&\frac{a_1(a_1+1-\mu_0)(1-s)+a_2(2a_1+1-\mu_0-a_\mu)s}{-s(a_2-a_1)\left[a_1\left(1-s\right)+a_2\left(1+s\right)\right]},\\
				p_3(a_1,a_2)&:=&\frac{\left(1-s\right)\left[sa_\mu+\left(1-s\right)a_1+(1-\mu_0)\right]}{-2s\left[a_1\left(1-s\right)+a_2\left(1+s\right)\right]}, \\
				p_1(a_1,a_2)&:=&1-p_2(a)-p_3(a).
			\end{eqnarray*}
			Observe that $p_1(a)+p_2(a)+p_3(a)=1$ and (by direct, tedious computation) $\sum_{i=1}^3 p_i(a)(a_i,\z_i) = (a_\mu,\mu_0)$.
			The affine independence property of $\{t_i(a)\}_{i=1}^3$ then tells us that $p^{\mathcal{T}}_{t(a)}= \sum_{i=1}^3 p_i(a) \delta_{\nu_{t_i(a)}}$.
			
			Now, the principal objective can be expressed (by direct, more tedious computation) as:\footnote{The ``$\mathbb V_{\theta\sim\mu}(\theta)+a_\mu^2$'' term is an irrelevant constant, which we add for convenience.  The first equality follows directly from Observation \ref{obs:action}.} 
			\begin{align}\label{eqn:affine}
			\int_\DT U_P \ddd p^{\mathcal{T}}_{t(a_1,a_2)} +\mathbb V_{\theta\sim\mu}(\theta)+a_\mu^2 &=  \mathbb E_{\nu\sim p^{\mathcal{T}}_{t(a_1,a_2)}}[a^*(\nu)^2] \nonumber\\
			&= p_1(a_1,a_2) a_1^2 + p_2(a_1,a_2) a_2^2 + p_3(a_1,a_2) \left( -a_1-\tfrac{2s}{1-s}a_2 \right)^2 \nonumber\\
			&= a_1^2-\frac{a_1^2+(1-\mu_0)a_1}{s}-\left(2a_1+\frac{1-\mu_0}{1-s}+\frac{2s-1}{1-s}a_\mu\right)a_2.
			\end{align}
			
			\item  If $(a_1,a_2)\in \mathcal A$ and $t_2(a_1,a_2)= t_3(a_1,a_2)$ (i.e. $a_2= -\frac{1-s}{1+s}a_1$),
			that $(a_\mu,\mu_0)\in \co\{t_1(a),t_2(a)\}$, $t_1(a)\in \lt$ and $t_2(a)\in \rt$ imply that
			\begin{align*}
			(a_1,a_2)&=\left(-\frac{1-\mu_0+sa_\mu}{1-s},\frac{1-\mu_0+sa_\mu}{1+s}\right), \\
			p^{\mathcal{T}}_{t(a_1,a_2)}(\nu_{t_1})&= \frac{1-s}{2}\frac{1-\mu_0-a_\mu}{1-\mu_0+sa_\mu},\\ p^{\mathcal{T}}_{t(a_1,a_2)}(\nu_{t_2})&= \frac{1+s}{2}\frac{1-\mu_0+a_\mu}{1-\mu_0+sa_\mu}.
			\end{align*}
			
			The value of the principal's objective is now $$\int_\DT U_P \ddd p^{\mathcal{T}}_{t\left(-\frac{1-\mu_0+sa_\mu}{1-s},\frac{1-\mu_0+sa_\mu}{1+s}\right)} +\mathbb V_{\theta\sim\mu}(\theta)+a_\mu^2
			= p_1 \left(-\tfrac{1-\mu_0+sa_\mu}{1-s}\right)^2 + p_2 \left(\tfrac{1-\mu_0+sa_\mu}{1+s}\right)^2 = \tfrac{(1-\mu_0)^2-s^2a_\mu^2}{1-s^2},$$
			which one can directly verify is consistent with the value of equation \ref{eqn:affine} at $(a_1,a_2)=\left(-\frac{1-\mu_0+sa_\mu}{1-s},\frac{1-\mu_0+sa_\mu}{1+s}\right)$.
		\end{itemize}
		
		Therefore, equation \ref{eqn:affine} summarizes the principal's payoff for all $(a_1,a_2)\in \mathcal{A}$. Observe that this objective is affine in $a_2$.  But, $t(\cdot)$ being continuous and $\mathcal T$ being compact, the set of $a_2$ such that $(a_1,a_2)\in \mathcal A$ is (for fixed $a_1$) a compact set of real numbers.  We may therefore find a principal-optimal attention outcome by restricting attention to the case that $a_2$ is the largest or smallest possible number for which $(a_1,a_2)\in \mathcal A$.  Letting $\mathcal A^*\subseteq \mathcal A$ be the set of pairs with this property, we can view the principal's problem as $$\max_{(a_1,a_2)\in \mathcal A^*} \int_\DT U_P \ddd p^{\mathcal{T}}_{t(a_1,a_2)}.$$
		
		What does $p^{\mathcal{T}}_{t(a)}$ look like if $(a_1,a_2)\in \mathcal A^*$?  As $t(\cdot)$ is continuous, any $(a_1, \z_1, a_2, \z_2, a_3, \z_3) \in \mathcal T$ with $(a_3,\z_3)\in \mbox{relint}(\rt)$, $(a_2,\z_2) \in \mbox{int}\left(\B\right)$ and $\mu_0\in \mbox{int}\left(\co\{(a_i,\z_i)\}_{i=1}^3\right)$ cannot have $(a_1,a_2)\in \mathcal A^*$.  The reason is that, from the definition of $\mathcal T$, it would then contain $t(a_1,a_2\pm\epsilon)$ for sufficiently small $\epsilon$. So, if $(a_1,a_2)\in \mathcal A^*$, then  $(a_3,\z_3)=(1,0)$,\footnote{It cannot be that $(a_3,\z_3)=(0,1)$ because $a_\mu\geq 0$ and $(a_\mu,\mu_0)\in \mbox{int}(\B)$.} or $(a_2,\z_2)$ belongs to the boundary of $\B$, or $\mu_0$ belongs to the boundary of $\co\{(a_i,\z_i)\}_{i=1}^3$.  
		
		If $(a_3,\z_3)=(1,0)$, then we have established part (ii) of this lemma.
		
		If $\mu_0$ belongs to the boundary of $\co\{(a_i,\z_i)\}_{i=1}^3$, it must be that $p_i(a_1,a_2)=0$ for some $i$, then $p^{\mathcal{T}}_{t(a)}$ has binary support.
		
		If $(a_2,\z_2)$ belongs to the boundary of $\B$, since $a_1< a_2$ and $0<s<1$, it cannot be that $a_2\in \lt\cup\down$.  Since $-s\neq -1$, it can only be that $(a_2,\z_2)\in\rt$ if $a_2=a_3=-a_1-\tfrac{2s}{1-s}a_2$, i.e. if $a_2=-\tfrac{1-s}{1+s}a_1$. But then $p^{\mathcal{T}}_{t(a)}$ also has binary support.
		
		So any $(a_1,a_2)\in \mathcal A^*$ has $p^{\mathcal{T}}_{t(a)}$ either binary, or ternary with $\nu_{(1,0)}$ in the support (and critical slopes), and the result follows.
		
		Finally, for any $t\in \mathbb{R}^6$ s.t. $(a_3,z_3)=(1,0)$, $(a_1,\z_1)\in \lt$ and $\frac{\z_2-\z_1}{a_2-a_1}=s=\frac{z_2}{1-a_2}$, we have $(a_\mu,\mu_0)\in \mbox{int}\left[\co\{(a_1,z_1),\ (a_2,z_2),\ (a_3,z_3)\}\right]$ if and only if $s> \frac{\mu_0}{1-a_\mu}$. So if $s\leq \frac{\mu_0}{1-a_\mu}$, either $t\notin \mathcal{T}$ or $p_t^\mathcal{T}$ is binary; while if $s> \frac{\mu_0}{1-a_\mu}$, $p_t^\mathcal{T}$ is of the second form in this lemma. (See Figures \ref{fig:bigsvaryx2} and \ref{fig:smallsvaryx2} for illustration.)  The last part of Lemma \ref{lem:binorsym} follows.
	\end{proof}	

	\bigskip
	
	Graphically, when $s>\frac{\mu_0}{1-a_\mu}$, for any given $a_1$, varying $a_2$ leads to policies depicted in Figure \ref{fig:bigsvaryx2} (following the notation of Lemma \ref{lem:binorsym}'s proof). On the other hand, when $s\leq \frac{\mu_0}{1-a_\mu}$, any $(a_1,a_2)\in \mathcal A^*$ leads to a binary policy, as depicted in Figure \ref{fig:smallsvaryx2}. 
	\begin{figure}[ht!]
		\centering
		\begin{tikzpicture}[scale=0.8]
		\draw (-2,0) -- (2,0) -- (0,2)--(-2,0);
		\node [above] at (0.3,0.65) {$\mu$};
		\node [below] at (2.8,0) {$a$};
		\node [right] at (0,2.8) {$z$};
		\node [below] at (0,-0.2) {$(a_1,a_2)\in \mathcal A^*$};
		\draw[>=stealth, dashed,->]  (0,0) -- (0,3);
		\draw[>=stealth, dashed,->]  (-2.7,0) -- (3,0);
		\draw[fill] (0.3,0.7) circle [radius=0.05];
		\draw[fill,red] (2,0) circle [radius=0.05];
		\draw[fill,red] (-0.02,1.6) circle [radius=0.05];
		\draw[fill,red] (-1.8,0.2) circle [radius=0.05];
		\draw [red] (-1.8,0.2) -- (-0.02,1.6) -- (2,0)-- (-1.8,0.2);
		\draw [dashed,blue] (0.0978255,1.7182604) -- (1.0173913,0.9826078);
		\draw [dashed,blue] (0.0489123,1.6791302) -- (1.4086956,0.5913039);
		\end{tikzpicture}
		\begin{tikzpicture}[scale=0.8]
		\draw (-2,0) -- (2,0) -- (0,2)--(-2,0);
		\node [above] at (0.3,0.65) {$\mu$};
		\node [below] at (2.8,0) {$a$};
		\node [right] at (0,2.8) {$z$};
		\node [below] at (0,-0.2) {$(a_1,a_2)\notin \mathcal A^*$};
		\draw[>=stealth, dashed,->]  (0,0) -- (0,3);
		\draw[>=stealth, dashed,->]  (-2.7,0) -- (3,0);
		\draw[fill] (0.3,0.7) circle [radius=0.05];
		\draw[fill,red] (0.0489123,1.6791302) circle [radius=0.05];
		\draw[fill,red] (-1.8,0.2) circle [radius=0.05];
		\draw[fill,red] (1.4086956,0.5913039) circle [radius=0.05];
		\draw [red] (-1.8,0.2) -- (0.0489123,1.6791302) -- (1.4086956,0.5913039)-- (-1.8,0.2);
		\draw [dashed,blue] (0.0978255,1.7182604) -- (1.0173913,0.9826078);
		\draw [dashed,blue] (2,0) -- (-0.02,1.6);
		\end{tikzpicture}
		\begin{tikzpicture}[scale=0.8]
		\draw (-2,0) -- (2,0) -- (0,2)--(-2,0);
		\node [above] at (0.3,0.65) {$\mu$};
		\node [below] at (2.8,0) {$a$};
		\node [right] at (0,2.8) {$z$};
		\node [below] at (0,-0.2) {$(a_1,a_2)\in \mathcal A^*$};
		\draw[>=stealth, dashed,->]  (0,0) -- (0,3);
		\draw[>=stealth, dashed,->]  (-2.7,0) -- (3,0);
		\draw[fill] (0.3,0.7) circle [radius=0.05];
		\draw[fill,red] (0.0978255,1.7182604) circle [radius=0.05];
		\draw[fill,green] (-1.8,0.2) circle [radius=0.05];
		\draw[fill,green] (1.1,0.9) circle [radius=0.05];
		\draw [red] (-1.8,0.2) -- (0.0978255,1.7182604) --  (1.1,0.9);
		\draw [green] (-1.8,0.2) --   (1.1,0.9);
		\draw [dashed,blue] (0.0489123,1.6791302) -- (1.4086956,0.5913039);
		\draw [dashed,blue] (2,0) -- (-0.02,1.6);
		\end{tikzpicture}
		\caption{When $s>\frac{\mu_0}{1-a_\mu}$, for a fixed $a_1$, varying $a_2$ such that $(a_1,a_2)\in \mathcal A.$}
		\label{fig:bigsvaryx2}
	\end{figure}
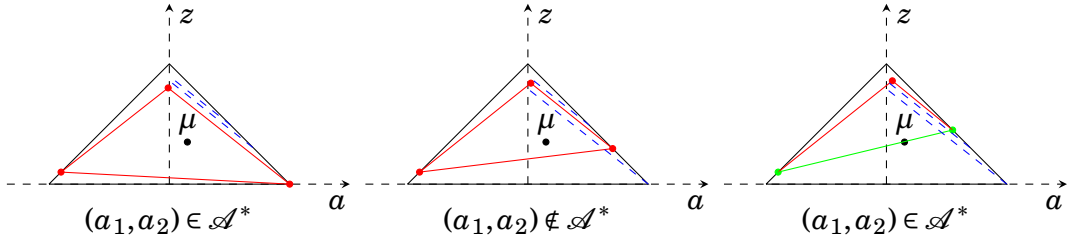
	
	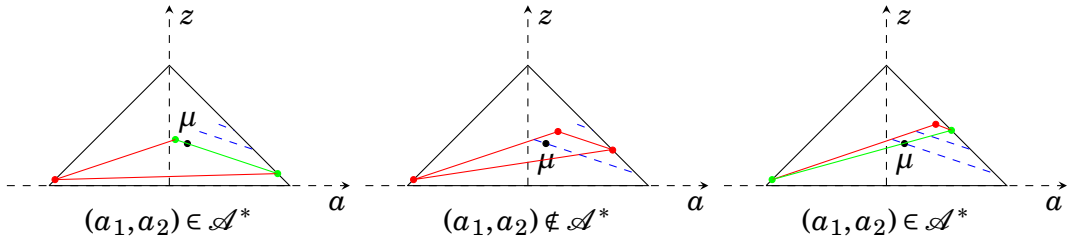
\begin{figure}[ht!]
		\centering
		\begin{tikzpicture}[scale=0.8]
		\draw (-2,0) -- (2,0) -- (0,2)--(-2,0);
		\node [above] at (0.3,0.7) {$\mu$};
		\node [below] at (2.8,0) {$a$};
		\node [right] at (0,2.8) {$z$};
		\node [below] at (0,-0.2) {$(a_1,a_2)\in \mathcal A^*$};
		\draw[>=stealth, dashed,->]  (0,0) -- (0,3);
		\draw[>=stealth, dashed,->]  (-2.7,0) -- (3,0);
		\draw[fill] (0.3,0.7) circle [radius=0.05];
		\draw[fill,green] (1.8,0.2) circle [radius=0.05];
		\draw[fill,green] (0.1,0.766666666) circle [radius=0.05];
		\draw[fill,red] (-1.9,0.1) circle [radius=0.05];
		\draw [red] (-1.9,0.1) -- (0.1,0.766666666) ;
		\draw [red] (1.8,0.2)-- (-1.9,0.1);
		\draw [green]  (0.1,0.766666666) -- (1.8,0.2);
		\draw [dashed,blue] (0.5,0.9) -- (1.4,0.6);
		\draw [dashed,blue] (0.82,1.02) --  (1.08,0.92);
		\end{tikzpicture}
		\begin{tikzpicture}[scale=0.8]
		\draw (-2,0) -- (2,0) -- (0,2)--(-2,0);
		\node [below] at (0.3,0.7) {$\mu$};
		\node [below] at (2.8,0) {$a$};
		\node [right] at (0,2.8) {$z$};
		\node [below] at (0,-0.2) {$(a_1,a_2)\notin \mathcal A^*$};
		\draw[>=stealth, dashed,->]  (0,0) -- (0,3);
		\draw[>=stealth, dashed,->]  (-2.7,0) -- (3,0);
		\draw[fill] (0.3,0.7) circle [radius=0.05];
		\draw[fill,red] (0.5,0.9) circle [radius=0.05];
		\draw[fill,red] (-1.9,0.1) circle [radius=0.05];
		\draw[fill,red] (1.4086956,0.5913039) circle [radius=0.05];
		\draw [red] (-1.9,0.1) -- (0.5,0.9) -- (1.4,0.6)-- (-1.9,0.1);
		\draw [dashed,blue]  (0.1,0.766666666) -- (1.8,0.2);
		\draw [dashed,blue] (0.82,1.02) --  (1.08,0.92);
		\end{tikzpicture}
		\begin{tikzpicture}[scale=0.8]
		\draw (-2,0) -- (2,0) -- (0,2)--(-2,0);
		\node [below] at (0.3,0.7) {$\mu$};
		\node [below] at (2.8,0) {$a$};
		\node [right] at (0,2.8) {$z$};
		\node [below] at (0,-0.2) {$(a_1,a_2)\in \mathcal A^*$};
		\draw[>=stealth, dashed,->]  (0,0) -- (0,3);
		\draw[>=stealth, dashed,->]  (-2.7,0) -- (3,0);
		\draw[fill] (0.3,0.7) circle [radius=0.05];
		\draw[fill,red] (0.82,1.02) circle [radius=0.05];
		\draw[fill,green] (-1.9,0.1) circle [radius=0.05];
		\draw[fill,green] (1.08,0.92) circle [radius=0.05];
		\draw [red] (-1.9,0.1) -- (0.82,1.02) --  (1.08,0.92);
		\draw [green] (-1.9,0.1) --   (1.08,0.92);
		\draw [dashed,blue]  (0.1,0.766666666) -- (1.8,0.2);
		\draw [dashed,blue] (0.5,0.9) -- (1.4,0.6);
		\end{tikzpicture}
		\caption{When $s\leq\frac{\mu_0}{1-a_\mu}$, for a fixed $a_1$, varying $a_2$ such that $(a_1,a_2)\in \mathcal A.$}
		\label{fig:smallsvaryx2}
	\end{figure}
	
	\hspace{-0.6cm}\emph{Proof of Proposition \ref{prop:optimal}. }
		We will find a nonredundant optimal attention outcome, which exists by Lemma \ref{thm:finite}. 
		Note that a nonredundant $p$ is an optimal attention outcome, if and only if it is IC and principal-optimal, if and only if $(p,p)$ is a solution to \ref{eqn:designer}.
		
		If $\kappa\leq \kappa_1=\frac{1}{2}$, then $s^*(\kappa)\geq1$. Note that the support $\supp(p^F)$ of the full disclosure policy is $\{\nu_{(-1,0)},\nu_{(0,1)},\nu_{(1,0)}\}$.
		So by Proposition \ref{prop:IC} and condition \ref{eqn:lineIC3states}, $p^F$ is IC. Since $p^F\succeq^B q$ for all $q\in \mathcal{R}(\mu)$ and $U_P(\nu)$ is convex, by Jensen's inequality $\left(p^F,p^F\right)$ is a solution to \ref{eqn:designer}, so $p^F$ is an optimal attention outcome, and part (\ref{optimal1}) of Proposition \ref{prop:optimal} follows.
		
		If $\kappa>\kappa_4=2$, by equation \ref{eqn:lineIC3states} any nonredundant information policy with more than one messages is not IC. Therefore, the ``no information'' policy $p^N$ is the only nonredundant IC policy, and so is the solution to \ref{eqn:designer1}. Then by Lemma \ref{thm:finite}, $\left(p^N,p^N\right)$ is a solution to \ref{eqn:designer}, so $p^N$ is an optimal attention outcome, and part (\ref{optimal5}) of Proposition \ref{prop:optimal} follows.
		
		If $\kappa=\kappa_4=2$, then $s^*(\kappa)=0$.  Given equation \ref{eqn:lineIC3states}, then, any nonredundant $p\in\RP(\mu)$ must have $p\{\nu_{(a,\z)}:\ (a,\z)\in B, \ \z=\mu_0\}$.  But $p^O$ is IC too, and $p^O\succeq^B p$ for such $p$.  As $U_P$ is convex, it follows that $p^O$ is an optimal attention outcome.  For this value of $\kappa$, the content of part (\ref{optimal4}) of Proposition \ref{prop:optimal} follows.
		
		Henceforth, we consider the remaining case that $\kappa \in (\kappa_1, \kappa_4)=\left(\tfrac12, 2\right)$, collectively covering parts (\ref{optimal2}), (\ref{optimal3}), and (\ref{optimal4})---excluding the special case of $\kappa=\kappa_4$---of the proposition.  Here, $s:= s^*(\kappa)\in (0,1)$.  In this case, Lemma \ref{lem:binorsym} applies, telling us that there is an optimal attention outcome $p^*$ that has one supported belief on each of $\lt$ and $\rt$, and is either ternary with support containing $\nu_{(1,0)}$ (with slopes between adjacent beliefs equal to $s$) or binary. 
		\begin{itemize}
			\item If $p^*$ is binary, then $p^*$ has support $\{\nu_{(a_1, \z_1)},\ \nu_{(a_2, \z_2)}\}$ for some distinct $(a_1,\z_1)\in\lt$ and $(a_2,\z_2)\in\rt$.  Let $\tilde s:= \tfrac{\z_2-\z_1}{a_2-a_1}$.  The fact that $(a_1,\z_1)\in\lt$ and $(a_2,\z_2)\in\rt$ implies that $\tilde{s}\in \left[-\frac{\mu_0}{1-a_\mu},\frac{\mu_0}{1+a_\mu}\right]$. We also know that $\tilde s\in [-s,s]$ by Lemma \ref{lem:optimalbinary}. 
			
			Given such $\tilde{s}$, 
			\begin{align*}
			(a_1,\z_1)&=\left(-\frac{1-\mu_0+\tilde{s}a_\mu}{1-\tilde{s}},\frac{\mu_0-\tilde{s}(1+a_\mu)}{1-\tilde{s}}\right)\text{ and }\\
			(a_2,\z_2)&=\left(\frac{1-\mu_0+\tilde{s}a_\mu}{1+\tilde{s}},\frac{\mu_0+\tilde{s}(1-a_\mu)}{1+\tilde{s}}\right)
			\end{align*}
			are the unique intersections of $\lt$ and $\rt$, respectively, with the line of slope $\tilde{s}$ through $(a_\mu,\mu_0)$. Next, Bayes-plausibility tells us that
			\begin{align*}
			p^*\{\nu_{(a_1, \z_1)}\}= \frac{1-\tilde{s}}{2}\frac{1-\mu_0-a_\mu}{1-\mu_0+\tilde{s}a_\mu}\\
			p^*\{\nu_{(a_2, \z_2)}\}= \frac{1+\tilde{s}}{2}\frac{1-\mu_0+a_\mu}{1-\mu_0+\tilde{s}a_\mu}
			\end{align*}

			Therefore, the principal's objective can be written as (appealing to Observation \ref{obs:action}) 
			\begin{align}\label{eqn:binaryfinal}
			\int_\DT U_P \ddd p^* &= \mathbb V_{\mu\sim p^*} \left[a^*(\nu)\right] - \mathbb V_{\theta\sim\mu}(\theta) \nonumber\\
			&= \int_\DT \left(a^*(\nu)-a_\mu\right)^2 \ddd p^*(\nu) - \mathbb V_{\theta\sim\mu}(\theta)\nonumber \\
			&= \tfrac{1-\tilde{s}}{2}\tfrac{1-\mu_0-a_\mu}{1-\mu_0+\tilde{s}a_\mu} \left(-\tfrac{1-\mu_0+a_\mu}{1-\tilde s}\right)^2 +  \tfrac{1+\tilde{s}}{2}\tfrac{1-\mu_0+a_\mu}{1-\mu_0+\tilde{s}a_\mu}\left(\tfrac{1-\mu_0-a_\mu}{1+\tilde s}\right)^2 - \mathbb V_{\theta\sim\mu}(\theta) \nonumber\\
			&= \tfrac{(1-\mu_0+a_\mu)(1-\mu_0+a_\mu)}{2(1-\mu_0+\tilde{s}a_\mu)} \left( \tfrac{1-\mu_0+a_\mu}{1-\tilde s} + \tfrac{1-\mu_0-a_\mu}{1+\tilde s} \right) - \mathbb V_{\theta\sim\mu}(\theta) \nonumber\\			
			&= \frac{(1-\mu_0)^2-a_\mu^2}{1-\tilde s^2} - \mathbb V_{\theta\sim\mu}(\theta).
			\end{align}

			As $0\leq a_\mu< 1-\mu_0$ (because $(a_\mu,\mu_0)\in \mbox{int}(\B)$), the above objective is strictly increasing in $|\tilde s|$, so $p^*$ being an optimal attention outcome implies that $|\tilde s| = \min\{s,\frac{\mu_0}{1-a_\mu}\}$.  
			
			\item If there is no binary optimal attention outcome, then $p^*$ is ternary with support $\{\nu_{(a_1, \z_1)},\ \nu_{(a_2, \z_2)},\ \nu_{(1, 0)}\}$, $(a_1,\z_1)\in \lt$ and $\frac{\z_2-\z_1}{a_2-a_1}=s=\frac{z_2}{1-a_2}$. By the last part of Lemma \ref{lem:binorsym}, this is only possible when $s>\frac{\mu_0}{1-a_\mu}$. Recall that 
			\begin{align*}
			p_2(a_1,a_2)&=\frac{a_1(a_1+1-\mu_0)(1-s)+a_2(2a_1+1-\mu_0-a_\mu)s}{-s(a_2-a_1)\left[a_1\left(1-s\right)+a_2\left(1+s\right)\right]}, \\
			t_3(a_1,a_2)&=(-a_1-\tfrac{2s}{1-s}a_2, \ 1+a_1+\tfrac{2s}{1-s}a_2).
			\end{align*} 
			So, that $t_3=(1,0)$ and $p_2(a_1,a_2)\geq 0$ imply that $a_2=-\frac{1-s}{2s}(1+a_1)$ and $a_1\in \left[-1,-\frac{1-\mu_0-a_\mu}{1-\mu_0+a_\mu}\right]$; and moreover, any such pair has $(a_1,a_2)\in\mathcal A$. Substituting into \ref{eqn:affine} yields
			\begin{align}\label{eqn:ternaryfinal}
			\int_\DT U_P \ddd p^* + \mathbb V_{\theta\sim\mu}(\theta)&=-a_\mu^2+a_1^2-\frac{a_1^2+(1-\mu_0)a_1}{s}+\left(2a_1+\frac{1-\mu_0}{1-s}+\frac{2s-1}{1-s}a_\mu\right)\frac{1-s}{2s}(1+a_1)\nonumber\\
			&=-a_\mu^2+\frac{1}{2s}\left[a_1(1-a_\mu+\mu_0-2s(1-a_\mu))+1-\mu_0-a_\mu+2sa_\mu\right].
			\end{align}
			Note that the above objective function is affine in $a_1$. Note also that when $a_1=-\frac{1-\mu_0-a_\mu}{1-\mu_0+a_\mu}$, $p_2(a_1,a_2)= 0$, so that the binary policy with slope $-\frac{\mu_0}{1-a_\mu}$ is obtained. Since, by hypothesis, no binary policy is an optimal attention outcome, $a_1$ must be optimally set to $-1$, i.e. $s> \frac{1-a_\mu+\mu_0}{2(1-a_\mu)}$.
		\end{itemize}

		Now, we consider various subcases for the value of $\kappa$, and find an optimal policy using payoffs computed in \ref{eqn:binaryfinal} and \ref{eqn:ternaryfinal}.  
		\begin{itemize}
			\item First, suppose $\kappa \in (\kappa_1,\kappa_2]$ so that $ \frac{1-a_\mu+\mu_0}{2(1-a_\mu)}\leq s<1$. In \ref{eqn:ternaryfinal}, it is optimal to set $a_1=-1$ because $s\geq \frac{1-a_\mu+\mu_0}{2(1-a_\mu)}$. But since setting $a_1=-\frac{1-\mu_0-a_\mu}{1-\mu_0+a_\mu}$ results in an optimal binary policy (i.e. a binary policy that maximizes the principal's utility among all IC binary policies),\footnote{Note that now $\min\{s,\frac{\mu_0}{1-a_\mu}\}=\frac{\mu_0}{1-a_\mu}$, and the resulting binary policy by setting $a_1=-\frac{1-\mu_0-a_\mu}{1-\mu_0+a_\mu}$ has slope $-\frac{\mu_0}{1-a_\mu}$.} Lemma \ref{lem:binorsym} then implies that the resulting ternary policy by setting $a_1=-1$ is an optimal attention outcome. This proves part (\ref{optimal2}) of the proposition. 
			\item Next, suppose $\kappa \in [\kappa_2,\kappa_3]$, so that $\frac{\mu_0}{1-a_\mu}\leq s\leq \frac{1-a_\mu+\mu_0}{2(1-a_\mu)}$.  In \ref{eqn:ternaryfinal}, it is optimal to set $a_1=-\frac{1-\mu_0-a_\mu}{1-\mu_0+a_\mu}$ because $s\leq \frac{1-a_\mu+\mu_0}{2(1-a_\mu)}$, which, as noted before, results in a binary policy. This means that any ternary policy with $\nu_{(1,0)}$ in the support (and critical slopes) is dominated by a binary policy. Then, by Lemma \ref{lem:binorsym}, we know that an optimal binary policy is optimal. That is, the binary policy with slope $|\tilde{s}|=\min\{s,\frac{\mu_0}{1-a_\mu}\}=\frac{\mu_0}{1-a_\mu}$ is an optimal attention outcome. This proves the part (\ref{optimal3}) of the proposition.
			\item Finally, suppose $\kappa \in(\kappa_3,\kappa_4)$ so that $0< s<\frac{\mu_0}{1-a_\mu}$.  Then by Lemma \ref{lem:binorsym}, we know that an optimal binary policy is optimal.\footnote{Note that the last part of Lemma \ref{lem:binorsym} says that  policies of the second form (ternary policies with $\nu_{(1,0)}$ in the support and critical slopes) do not exist when $s<\frac{\mu_0}{1-a_\mu}$, so we can only focus on binary policies.} That is, the binary policy with slope $|\tilde{s}|=\min\{s,\frac{\mu_0}{1-a_\mu}\}=s$ is an optimal attention outcome. This proves the remainder of part (\ref{optimal4}) of the proposition.
		\end{itemize}
		
		\medskip
		
		All that remains is the uniqueness result.  Let us show that there is a unique (up to reflection) optimal attention outcome whenever $\kappa\neq\kappa_2$.
		
		If $\kappa\leq\tfrac12$, then full information is IC, and it is the unique minimizer of $p\mapsto\int_{\DT}U_P\ddd p$ over $\RP(\mu)$, and so is the unique optimal attention outcome.  If $\kappa>2$, then no information is uniquely IC, and so is the unique optimal attention outcome.  If $\kappa=2$, then every IC information policy is supported on $\{\nu_{(a,\z)}:\ (a,\z)\in \B, \ \z=\mu_0\ \}$, a convex set over which $U_P$ is strictly convex; $p^O$ is therefore (being a unique $\succeq^B$-maximum over all such information policies) the unique optimal attention outcome.
		
		Now, we focus on the remaining case of $\tfrac12<\kappa<2$, and let $p^*$ be an arbitrary optimal attention outcome.
		
		Letting $\mathcal P^*$ be the set of nonredundant optimal attention outcome, Claims OA.7 and OA.9 of \cite{Lipnowski2020} tell us that $p^*\in \overline{\text{co}}(\mathcal P^*)$. 
		As $U_A$ is strictly concave on $\lt$ (since $s<1$) and $p^*$ is IC, it must be that the positive measures $p^*(\cdot\cap \lt)$ and $p^*(\cdot\cap \rt)$ both have support of size at most one.  But, by Lemma \ref{lem:binorsym}, any element of $\overline{\text{co}}(\mathcal P^*)\setminus \mathcal P^*$ would violate this property.  Therefore $p^*\in\mathcal P^*$, i.e. it is nonredundant. 
		
		We now show that $p^*$ is unique (up to reflection) unless $\kappa=\kappa_2$. To this end, we need to rule out both multiplicity within the class of policies of the form guaranteed by Lemma \ref{lem:binorsym}, and existence of an optimal nonredundant policy outside of that class. The first part of the current proof already shows that within that class, the optimal attention outcome is unique unless $\kappa=\kappa_2$. We now argue that, as long as $\kappa\neq \kappa_2$, no nonredundant policy outside that class can be an optimal attention outcome. Following the notation of Lemma \ref{lem:binorsym}'s proof, take any $(a_1,a_2)\in \mathcal{A} \setminus \mathcal{A}^*$, and assume, for a contradiction, that $p^{\mathcal{T}}_{t(a)}$ is an optimal attention outcome. As $\tilde a_2 \mapsto \int_\DT U_P\ddd p^{\mathcal{T}}_{t(a_1,\tilde a_2)}$ was shown to be affine, this can only be true if the affine function is constant, so that $p^{\mathcal{T}}_{t(a)}$ generates the same payoff to the principal as $p^{\mathcal{T}}_{t(a_1,\underline{a}_2)}$ and $p^{\mathcal{T}}_{t(a_1,\bar{a}_2)}$ do, where $(a_1,\underline{a}_2),(a_1,\bar{a}_2)\in\mathcal{A}^*$ with $\underline{a}_2<a_2<\bar{a}_2$. In particular, both $p^{\mathcal{T}}_{t(a_1,\underline{a}_2)}$ and $p^{\mathcal{T}}_{t(a_1,\bar{a}_2)}$ are also optimal attention outcomes.  Consider alternative cases of $\kappa \in (\tfrac12, 2)\setminus\{\kappa_2\} = (\kappa_1,\kappa_2)\cup (\kappa_2,\kappa_3)\cup [\kappa_3,\kappa_4)$.

		\begin{itemize}
			\item If $\kappa\in (\kappa_1,\kappa_2)\cup (\kappa_2,\kappa_3)$, then $p^{\mathcal{T}}_{t(a_1,\underline{a}_2)}$ is ternary and $p^{\mathcal{T}}_{t(a_1,\bar{a}_2)}$ is binary.  But, as we argued above, the optimal ternary attention outcome and the optimal binary attention outcome are strictly payoff ranked for $\kappa\neq\kappa_2$. This is a contradiction to $p^{\mathcal{T}}_{t(a_1,\underline{a}_2)}$ and $p^{\mathcal{T}}_{t(a_1,\bar{a}_2)}$ both being optimal attention outcomes.
						
			\item If $\kappa\in [\kappa_3,\kappa_4)$, then both $p^{\mathcal{T}}_{t(a_1,\underline{a}_2)}$ and $p^{\mathcal{T}}_{t(a_1,\bar{a}_2)}$ are binary, while one element in $\supp(p^{\mathcal{T}}_{t(a_1,\underline{a}_2)})$ is in $\mbox{int}(\B)$ (see the left panel of Figure \ref{fig:smallsvaryx2} for illustration). By Claim \ref{cl11}, $p^{\mathcal{T}}_{t(a_1,\underline{a}_2)}$ is not principal-optimal, a contradiction to $p^{\mathcal{T}}_{t(a_1,\underline{a}_2)}$ being an optimal attention outcome.{\scriptsize$\blacksquare$}
		\end{itemize}

\end{spacing}

\end{document}